\documentclass[11pt]{article}
\usepackage[margin=1in]{geometry}
\usepackage{amssymb,amsthm}
\usepackage{booktabs} 

\usepackage{mathtools}
\usepackage{subfig}
\usepackage{float}

\usepackage{url}            

\usepackage{tcolorbox}
\usepackage{bbm}
\usepackage{dsfont}
\usepackage{amsmath,amsfonts,amssymb,commath,hyperref}
\usepackage{natbib}
\usepackage{caption}
\usepackage[colorinlistoftodos,prependcaption,textsize=small]{todonotes}
\usepackage[shortlabels]{enumitem}
\usepackage{xspace}
\usepackage[multiple]{footmisc}

\usepackage{algorithm}
\usepackage[noend]{algpseudocode}

\usepackage[capitalize]{cleveref}

\usepackage[capitalize]{cleveref}
\crefname{claim}{Claim}{Claims}

\usepackage[suppress]{color-edits}
\addauthor{ch}{black}
\addauthor{df}{orange}
\addauthor{r}{magenta}
\addauthor{rev}{blue}

\newcommand{\detPi}{\widetilde{\Pi}}
 
\newcommand{\bx}{\mathbf{x}}
\newcommand{\by}{\mathbf{y}}
\newcommand{\bz}{\mathbf{z}}

\newcommand{\randomstatevector}{{\mathbf{N}}}

\newcommand{\randomstate}{{N}}
\newcommand{\scaledrandomstate}{N^\theta}

\newcommand{\rewardset}{\Xi}

\newcommand{\rev}{R}
\newcommand{\workerid}{w}
\newcommand{\xvec}{\mathbf{x}}
\newcommand{\simplex}{\boldsymbol{\Delta}}
\newcommand{\EE}{\mathbb{E}}
\newcommand{\fluidopt}{\textsc{Fluid-Opt}}
\newcommand{\supplyopt}{\textsc{Supply-Opt}}
\newcommand{\scaledlambda}{\lambda^\theta}
\newcommand{\scaledv}{v_{\theta}}
\newcommand{\PP}{\mathbb{P}}
\newcommand{\EEn}{\EE\left[\randomstate^\star\right]}
\newcommand{\Lambdaub}{\overline{\Lambda}}
\newcommand{\Lambdalb}{\underline{\Lambda}}
\newcommand{\supplyoptsol}{\bx^S}
\newcommand{\fluidn}{\widetilde{N}}
\newcommand{\avgreward}{\widehat{r}}
\newcommand{\supp}{\text{supp}}
\newcommand{\nf}{\normalfont}
\newcommand{\rmax}{r_{\max}}
\newcommand{\rmin}{r_{\min}}
\newcommand{\avgloss}{\hat{\ell}}
\newcommand{\algxvec}{\widetilde{\bx}^S}
\newcommand{\algx}{\widetilde{x}^S}

\newcommand{\scaledprofit}{\widehat{\Pi}^\theta}

\newcommand{\optfluidprofit}{\widetilde{\Pi}^*}

\newcommand{\unitvec}{\mathbf{e}}

\newcommand\blfootnote[1]{%
  \begingroup
  \renewcommand\thefootnote{}\footnote{#1}%
  \addtocounter{footnote}{-1}%
  \endgroup
}
\usepackage{natbib}
\usepackage{adjustbox}

\usepackage{authblk}
\usepackage{booktabs} 

\allowdisplaybreaks
\newcommand{\Halmos}{{}}
\newtheorem{theorem}{Theorem}
\numberwithin{theorem}{section}
\newtheorem{definition}[theorem]{Definition}

\newtheorem{remark}[theorem]{Remark}
\newtheorem{lemma}[theorem]{Lemma}

\newtheorem{proposition}[theorem]{Proposition}

\newtheorem{assumption}{Assumption}

\usepackage{setspace}
\begin{document}
	\title{Fair Incentives for Repeated Engagement\blfootnote{A preliminary version of this paper appeared as an extended abstract in WINE 2023.}}
        \author[1]{Daniel Freund}
        \author[2]{Chamsi Hssaine}
        \affil[1]{Massachusetts Institute of Technology, Cambridge, MA}
        \affil[2]{University of Southern California, Marshall School of Business, Los Angeles, CA}

	\date{}
	\maketitle

	\begin{abstract}
            We study a decision-maker's problem of finding optimal monetary incentive schemes for retention when faced with agents whose participation decisions {(stochastically) depend on the incentive they receive}. Our focus is on policies constrained to fulfill two fairness properties that {preclude outcomes wherein} different groups of agents experience different treatment on average. {We formulate the problem as a high-dimensional stochastic optimization problem, and study it through the use of a closely related deterministic variant.} We show that the optimal \emph{static} solution to this deterministic variant is asymptotically optimal for the {\it dynamic} problem under fairness constraints. Though solving for the optimal static solution gives rise to a non-convex optimization problem, we uncover a structural property that allows us to design a tractable, fast-converging heuristic policy.
Traditional schemes for retention ignore fairness constraints; indeed, the goal in these is to use differentiation to incentivize {repeated engagement with the system}. Our work $(i)$ shows that even in the absence of explicit discrimination, {dynamic policies may unintentionally discriminate between agents of different types by varying the type composition of the system}, and $(ii)$ presents an asymptotically optimal policy to avoid such discriminatory outcomes.
	\end{abstract}

\section{Introduction}\label{sec:intro}

Stakeholder retention is a fundamental challenge faced by many organizations.
For instance, it was said of former IBM executive Buck Rodgers that he ``behaved as if every IBM customer
were on the verge of leaving and that
[he'd] do anything to keep them from
bolting'' \citep{rodgers1986ibm}. Indeed, though it is conventional wisdom that continued {\it growth}  is necessary to ensure the success of a business, studies have routinely found that increasing customer {\it retention} rates by as little as 5\% could lead to an increase in profits of up to 95\% \citep{gallo2014value}. 

Of the many ways in which an organization can increase retention, one important lever it has at its disposal is that of {unconditional} {\it monetary incentives}{, which are paid out regardless of the retention outcome}. {Such incentives are used in many practical contexts:}
\begin{itemize}
    \item \textbf{Customer-side retention by for-profit corporations}: 
{Many loyalty programs use monetary rewards, e.g., gift cards in exchange for redeemable points, to retain customers. When the redemption rewards vary with time and  points need to be redeemed at an arbitrary point of time (e.g., after a certain point balance has been achieved but before they expire), the customer's reward can be modeled as, effectively, random. A concrete example of this arises in luxury condos that offer reward programs to their tenants \citep{gables23captivate}. Though tenants receive loyalty points every month, points can only be redeemed when the balance reaches a certain point. In addition, points need to be redeemed quickly, lest they expire, and their redemption is independent of a lease renewal decision, i.e., they are unconditional.}
    \item \textbf{Non-profit organizations}: Lotteries have been found to be extremely effective in incentivizing charitable giving in empirical studies \citep{landry2006toward}. In practice, certain banks have set up programs, referred to in Germany as \emph{Gewinnsparen}, that enroll their clients for recurring donations for local causes in exchange for a chance to win a cash amount \citep{gewinnsparen2}. 
    Monetary awards have also been found to effectively combat against volunteer attrition, which plagues non-profit organizations across the board \citep{downs2014modeling,frey2017volunteer}.
\end{itemize}

{In implementing these sorts of monetary incentives, one important concern that a decision-maker faces is that of {\it fairness}, especially in light of the abundance of examples in which algorithms deployed in the real world unintentionally discriminate against protected groups~\citep{kleinberg2018algorithmic}. {
Within the context of monetary incentives for retention, 
a decision-maker may want  
to maximize her bang-per-buck by paying individuals with different earnings sensitivities different amounts in order to to minimize the cost of retaining them. 
{However, s}ince individuals' earnings sensitivities correlate with protected classes such as gender \citep{heckert2002gender}, such unconstrained policies would discriminate between classes.}

{Our work aims to better understand} {the role of fairness in the} optimal design of {unconditional} monetary incentives for repeated engagement. Indeed, despite the frequent use of such incentives in practice, as well as extensive empirical work on their effectiveness in a variety of settings, to the best of our knowledge there have been few attempts to develop {\it theoretical} insights into their design {and the associated risk of algorithmic group discrimination.} We further detail our contributions below.
}

\subsection{Summary of contributions}

We consider a model wherein agents join a system in each (discrete-time) period, and receive a (possibly random) reward to remain in the system in the next period. At the end of the period, {\it unaware of the underlying distribution} from which rewards are drawn, agents probabilistically make {a} decision based on the reward received in the period to stay in the system, or leave once and for all. Specifically, we assume that agents are partitioned into {\it types} defined by $(i)$ their sensitivity to rewards, formalized via a {\it departure function} {that maps rewards received to the probability of departing}, and $(ii)$ the rate at which they join the system.\footnote{In Appendix \ref{apx:fairness-thm} we extend our main results to a rational-entry model wherein the types also encapsulate a reservation value wherein arrivals of a type only occur if their average reward exceeds the reservation value.} This model, though simple, gives rise to a spectrum of models of agent behavior. In particular, most of our results hold for {\it any} departure functions, as long as these  functions are non-increasing in the reward paid out to an agent.

The decision-maker collects some revenue associated with the number of agents in the system in each period, and incurs the cost associated with incentivizing these agents to stay according to the chosen reward distribution (where the support of this distribution is assumed to be {an arbitrary finite} set). The goal of the decision-maker is to determine the optimal policy to maximize her long-run average profit. As is common in classical stochastic control problems, the infinite-horizon Markov Decision Process (MDP) associated with the decision-maker's optimization problem suffers from the {curse of dimensionality}, which motivates the task of finding {\it near-optimal} policies.

One natural approach a decision-maker may want to take to maximize her profit is to {\it learn, then discriminate}: given the history of rewards paid out to each agent, the decision-maker could try to estimate each agent's type, and ``target'' agents whom she believes would stay in the system for lower rewards. {However,} not only is such {\it explicit} discrimination potentially problematic from a public relations standpoint, but it also runs counter to {\it  group fairness}, which at a high level requires that {an algorithm treat (reward) individuals belonging to different groups (e.g., demographic groups) similarly.\footnote{Within the field of machine learning, the notion of group fairness is typically defined via statistical parity, wherein a classifier must assign subjects in protected and unprotected groups to a class with equal probability.}} 
Avoiding this sort of {explicit} discrimination, the decision-maker can then turn to {\it dynamic} policies that draw rewards i.i.d. from the same distribution {\it in each period}, all the while actively managing the number of agents in the system by varying the distribution {\it across periods}. {However, even these seemingly fair policies may discriminate implicitly:} we show {that they} {can lead to some groups receiving consistently higher rewards than others due to the members of the former (latter) group \emph{self-selecting} into periods with higher (lower) rewards.}

{
{Against this backdrop, our work aims to {characterize} retention policies that fulfill two stringent fairness requirements:} $(i)$ agents {must} be paid from the same reward distribution {\it in each time period}, and $(ii)$ {agents of} different  types must experience the same reward distribution on average, over a long enough time horizon (\cref{def:fair_policy}). The two fairness constraints respectively enforce distributional envy-freeness {\it within} and {\it across} periods. On the one hand, if one views the reward distribution abstraction as agents playing a lottery for retention, distributional envy-freeness within periods ensures that agents know they are playing the same lottery. This {is a reasonable} principle given that, in the systems we consider, agents are symmetric from a revenue perspective (e.g., there is no ``specialization'' across types).
On the other hand, though group fairness may not be a requirement in all settings, the lack thereof may be problematic (for instance, in settings where departure probabilities are correlated with protected classes such as gender and race~\citep{heckert2002gender}). In particular, in such contexts an audit --- or even a company-authored DEI report --- might find in hindsight that{, despite avoiding explicit discrimination and not even attempting to learn the types of each agent, an organization may nonetheless pay} higher bonuses/rewards to some demographic than to another; whether or not this is justifiable, it poses a potential risk that organizations should be aware of. 
}

\chedit{In our first contribution, we 
show that there exists a static policy that is asymptotically optimal amongst the space of all policies that fulfill our fairness criteria \chedit{(\cref{thm:static-policies-are-opt-for-one-type,thm:asymptotic-result})}. We also show that dynamic policies can strictly outperform any static policy \chedit{(\cref{prop:explicit-wage-disc,ex:steady-state-cyclic})}; in other words, the asymptotic value of dynamic policies in our setting arises solely from the ability to (implicitly) discriminate between types.
Though we do not seek to prescribe {our stringent fairness constraints} for all retention settings,
a main insight of our work is that the value of a dynamic policy, in our setting, comes only from exploiting different retention probabilities across groups.}
 
{The proof that a static policy is asymptotically optimal is constructive, i.e., we design a heuristic fluid-based static policy for which our asymptotic guarantees hold as the market size is scaled by a parameter $\theta$. We moreover show {(\cref{thm:asymptotic-result})} that our fluid-based heuristic is {\it fast-converging} in the sense that it converges to a fluid upper bound at a rate of $\mathcal{O}(\frac{1}{\theta})$  (under a mild technical condition on the decision-maker's revenue function; without this condition it converges at a rate of $\mathcal{O}(\frac{1}{\sqrt{\theta}})$).} 
 
 {While our policy satisfies these natural desiderata,} computing it requires us to solve a high-dimensional nonconvex optimization problem 
 which is, \emph{a priori}, nontrivial to optimize. {In our final technical contribution}, we show a surprising structural property of the problem that allows us to efficiently compute its optimal solution. In particular, {\it independent of the size of the reward set, the number of agent types, and their departure probabilities}, there exists a fluid-optimal reward distribution that places positive weight on {\it at most} two rewards (\cref{thm:main-theorem}). This allows us to identify an optimal solution by considering all pairs of rewards, and then solving a KKT condition that consists of a single equation in one variable for each pair. \chedit{Similar two-reward structure of fluid-based policies has been identified in other operational settings \citep{bassamboo2016scheduling}; we use this result to derive insights into the structure of the fluid optimal policy for certain special cases. In particular, we show that the convexity of the departure probability function impacts the optimal dispersion level of the optimal reward scheme, lending credence to the use of ``surprise-and-delight'' lotteries for retention (\cref{thm:main-theorem2}).}

\paragraph{Structure of the paper.} 
{In Section~\ref{ssec:related-work}, we survey related literature. We then present the model and formulate the decision-maker's optimization problem in Section~\ref{sec:preliminaries}. We use a deterministic relaxation of our system to show in Section \ref{sec:main-results} that the fluid-based heuristic is optimal amongst all policies that satisfy our fairness constraints; though we also show the existence of dynamic policies that outperform the fluid heuristic, these policies are inherently discriminatory. Section~\ref{sec:asymptotic-opt} is devoted to analyzing the fluid-based heuristic and proving its fast convergence to the value of the fluid relaxation in a large-market regime. \chedit{Finally, Section \ref{sec:special-cases} leverages our analysis of the fluid heuristic to characterize optimal policies in special cases of interest.} 
All proofs are deferred to the appendix.}

\subsection{Related work}\label{ssec:related-work}

\paragraph{Workforce capacity planning.} Our work is related to the topic of {\it workforce capacity planning}, which has a long history in the operations management literature (see, e.g., \citet{de2015workforce} for an excellent survey). Within this line of work, we highlight papers that consider attrition and retention aspects of workforce planning. In contrast to our work, which focuses on the question of issuing monetary incentives throughout an agent's lifetime to retain them, these works are concerned with {hiring}, promotion, and {termination} decisions. For example, motivated by the naval aviation system, early work by \citet{grinold1976manpower} considered optimal accession policies when aviators have a known and deterministic lifetime. More recently, \citet{hu2016strategic} studied optimal hiring and admission and training policies for junior nurses, a profession in which attrition is pervasive, and as a result has been a central focus of much of the workforce planning literature. In their model, a {\it fixed} and exogenous fraction of the population leaves the system in each period, whereas in ours the decision-maker aims to set incentives in order to affect their retention. This work also resembles ours in that agents are {\it homogeneous} from a skills' perspective (though ours are heterogeneous with respect to their departures). 

A subset of the literature on workforce capacity planning is interested in worker {\it heterogeneity}; however, most of these works focus on heterogeneity with respect to skill set, not with respect to attrition. \citet{ahn2005staffing} consider a model in which workers turn over independently of the organization's policy and the state of the system, whereas \citet{gans2002managing} and \citet{arlotto2014optimal} allow workers' departure decisions to depend on the state of the system, but not the decision-maker's policy nor their own history in the system. Most recently, \citet{jaillet2021strategic} considered a more complex model of hiring, dismissing and promoting when workers' resignation decisions depend on their ``time-in-grade,'' or lifetime, in the system. {To the best of our knowledge, no works in this stream consider the fairness implications of policies.}

\paragraph{Customer retention.} We highlight the most closely related works here, as this area has a rich history in the marketing liteature (see, e.g., \citet{ascarza2018pursuit} for a survey). To the best of our knowledge, none of these papers focus on designing {\it fair} customer retention policies. On the contrary, the goal in these latter works is precisely to use differentiation in order to incentivize customers to stay. For example, in a computational study \citet{lemmens2020managing} define a profit-based loss function to predict, for each customer, the financial impact of a retention intervention, ranking customers based on the marginal impact of the intervention on churn, and post-intervention profits. \citet{aflaki2014managing} develop theoretical insights around optimal retention policies, in a setting where customer ``types,'' or sensitivities to interventions, are known by the decision-maker, thus allowing for customer differentiation; in their model, the optimal decision across the population decouples into optimal decisions for each individual customer. 
A separate stream of work investigates how {\it capacity} decisions affect service access quality, and customer retention as a result \citep{afeche2017customer,furman2021customer}. In contrast to the interventions we consider, which occur {\it in each period}, in the settings these latter papers consider, the decision-maker is constrained to make a single decision {\it at the beginning} of the time horizon.   

\paragraph{Empirical work on effectiveness of monetary incentives.} The effectiveness of monetary incentives for retention is also well-documented in the medical community. Empirical studies highlight their efficacy within the context of adherence to medication~\citep{volpp2008test,kimmel2012randomized}, weight loss and exercise~\citep{volpp2008financial,meeker2021combining}, postpartum compliance~\citep{stevens1994incentives}, and home-based health monitoring~\citep{sen2014financial}, for instance.

\paragraph{Algorithmic fairness.} Finally, our work adds to the large and growing body of work on the design of {\it fair} algorithms. We note that there is no universally agreed-upon notion of fairness (see, e.g., \citet{mehrabi2021survey} for a comprehensive overview of the different notions of fairness that have been considered in the literature). The notion of fairness that we choose to focus on is that of {\it group fairness} (as opposed to {\it individual fairness} \citep{dwork2012fairness}), which itself has no single definition. The one closest to ours, that arises in the machine learning literature within the context of group-fair classifiers, is {\it statistical parity} \citep{corbett2017algorithmic}. This requires that individuals in both  protected and non-protected groups have equal probabilities of being assigned to the positive predicted class; interpreting the set of rewards as possible predicted classes, our fairness definition can be understood as a multidimensional variant of statistical parity.

{An operational setting that is related to our study is fair (online) resource allocation;  this has received significant recent attention \citep{sinclair2021sequential,bateni2022fair,allouah2022robust, manshadi2021fair,banerjee2023online,freund2023group, freund2023good,balseiro2022uniformly}.} Other related literature on fair algorithms includes pricing problems with fairness constraints such as those studied by \citet{cohen2021dynamic,salem2021taming} and \citet{cohen2022price}. However, none of these works model customer attrition.  
	\section{Preliminaries}\label{sec:preliminaries}

We consider a discrete-time, infinite-horizon model of an organization{, which we henceforth generically refer to as a {\it system}}. {In each period {agents join the system, receive a reward, and decide whether to stay in the system for future periods or leave.} The {decision-maker} makes a profit, in each period, composed of the revenue from the {number of agents in the system}, net of {the cost of} the rewards {paid out to {agents}}. {For example, within the employment context, an agent corresponds to a worker performing a set of tasks in each period, with the reward corresponding to a bonus incentive; an agent can similarly correspond to a customer who enjoys service from her cable provider in a given period, with the reward corresponding to a discount.} We formalize each component of the model below}, {beginning with some technical notation.} {For clarity of exposition, we defer a lengthy discussion of modeling assumptions to the end of the section.}

{\paragraph{Technical notation.} Throughout the paper, $\mathbb{R}^+ = \left\{x \in \mathbb{R} \, | \, x \geq 0\right\}$, $\mathbb{R}^{> 0} = \left\{x \in \mathbb{R} \, | \, x > 0\right\}$, and $\mathbb{N}^+ = \left\{i \in \mathbb{N} \, | \, i \geq 1\right\}$. For $K  \in \mathbb{N}^+$, we let $[K] = \{1,\ldots,K\}$. We use $\supp(f)$ to denote the support of a given probability mass function $f$, i.e., $\supp(f) = \left\{x : f(x) > 0\right\}$, and $|\supp(f)|$ denotes the cardinality of its support. Moreover, given set $\rewardset \subset \mathbb{R}$, we let {$\Delta^{|\rewardset|}=\left\{\xvec\in[0,1]^{|\rewardset|}: \sum_r x_r=1\right\}$} be the standard probability simplex over $\rewardset$.  {Finally, $\unitvec_r$ denotes {the unit vector in the direction of reward~$r$.}}}

\subsection{Basic setup}\label{sec:basic-setup}

{\paragraph{Agents.}
We assume there are $K \in \mathbb{N}^+$ types of agents, defined by $(i)$ their reward sensitivity, and $(ii)$ the rate at which they join the system. Specifically, for $i \in [K]$, a type $i$ agent is associated with a {\it departure probability function} $\ell_i: \rewardset \mapsto (0,1]$, where $\rewardset \subset \mathbb{R}^+$} is a finite set of rewards from which the decision-maker chooses to compensate its agents.\footnote{{Heterogeneity in departure probability functions has also been considered in \citet{ovchinnikov2014balancing} in a simple model with two types.}} We assume that~$\ell_i$ is non-increasing and known to the decision-maker. Let $\rmin = \inf\{r \, | \, r \in \rewardset\}$ and $\rmax = \sup\{r \, | \, r \in \rewardset\}$. Moreover, the number of type $i$ arrivals in period $t$, $A_i(t)$, is drawn i.i.d. (across types and periods) from a $\text{Pois}(\lambda_i)$ distribution also known to the decision-maker; let~$\mathbf{A}(t) = \left(A_i(t), i = 1 \in [K]\right)$. 

\chedit{
\begin{remark}\label{remark:endo}
The exogenous arrival assumption is for ease of exposition. In Appendix \ref{apx:fairness-thm} we show that our main result applies to an endogenous entry model in which agents choose to enter the system by comparing their long-run average earnings to a type-dependent reservation value. 
\end{remark}
}

\paragraph{Periods.} Period $t \in \mathbb{N}$ is defined by the following sequence of events: $(i)$ for all $i \in [K]$, $A_i(t)$ type $i$ agents join the system; $(ii)$ each agent $\workerid$ is in the system (e.g., enjoying service, or working, for the duration of the period), and collects a (possibly random) reward $r_{\workerid}$; $(iii)$ having collected reward $r_{\workerid}$, {agent $\workerid$, unaware of the reward distribution, departs from the system with probability $\ell_{i_\workerid}(r_{\workerid})$, where $i_\workerid \in [K]$ denotes the type of agent $\workerid$;
otherwise, she remains in the system and moves onto the next period. We assume that the agent's decision to leave the system is made independently from all other agents, and independently of her prior history of rewards, i.e., agents make a decision based only on {their most recent experience}; moreover}, once an agent leaves, she does not return in a later period. ({This latter assumption --- that the agent is ``lost for good'' --- is standard in the workforce planning and customer retention literature; see, e.g., \citet{aflaki2014managing}, \citet{arlotto2014optimal}.})

We use $D_i(t)$ to denote the number of type $i$ departures at the end of period $t$, with $\mathbf{D}(t) = \left(D_i(t), i \in [K]\right)$. Finally, let $\randomstate_i(t)$ be the number of type $i$ agents in the system in period $t$ (and thus requiring payment), with $\randomstatevector(t) = \left(\randomstate_i(t), i \in [K]\right)$, and $\randomstate(t) = \sum_{i \in [K]} \randomstate_i(t)$. $\randomstatevector(t)$ {is based on the number of agents who were in the system in the previous period and did not depart, as well as the number of new agents who joined {at the beginning of the} current period. Given the described dynamics, the {\it state} of the system is fully characterized by} $\randomstatevector(t)$, which evolves as
\begin{align*}
\randomstatevector(t+1) = \randomstatevector(t) - \mathbf{D}(t) + \mathbf{A}(t+1), \qquad \forall \, t \in \mathbb{N}.
\end{align*}
{We next specify the decision-maker's objective and corresponding optimization problem.}

\paragraph{Objective.}

Given $\randomstate(t)$, the total number of agents in period $t$, the decision-maker obtains revenue $\rev\left(\randomstate(t)\right)$, where $\rev:\mathbb{R}^+ \mapsto \mathbb{R}^+$. We assume that $\rev$ is time-invariant, and depends only on the {\it total number} of agents in the system, rather than the type composition of the agent pool in each period. $\rev$ is moreover assumed to be $L$-Lipschitz continuous and differentiable over $\mathbb{R}^{+}$, as well as non-decreasing and concave. Given $\{r_\workerid\}_{\workerid = 1}^{\randomstate(t)}$, the (possibly random) set of rewards paid out to agents in period $t$, the period-$t$ profit is given by:
$$\Pi(t) = \rev\left(\randomstate(t)\right) - \sum_{\workerid=1}^{\randomstate(t)} r_\workerid.$$

Let $\widehat{\Pi}(t)$ denote the expected profit 
{in period} $t$, where the expectation is taken over the randomness in the reward realizations. 

{We formulate the decision-maker's optimization problem as a discrete-time, infinite horizon Markov Decision Process (MDP), where the objective is to maximize the long-run average profit.} Suppose the initial condition is $\randomstatevector(0) = \mathbf{c}_0$, $\mathbf{c}_0 \in \mathbb{N}^K$. For any policy $\varphi$, let $v({\varphi})$ denote the long-run average profit under $\varphi$ (assuming this limit exists). Formally:
\begin{align*}
v(\varphi) =& \lim_{T\to \infty} \frac1T \mathbb{E}\left[\,\sum_{t=1}^T\widehat{\Pi}(t)\, \bigg{\vert} \, \randomstatevector(0) = \mathbf{c}_0\,\right].
\end{align*}

{In complete generality, a policy $\varphi$ maps the set of  agents in the system in a given period, and the history of rewards observed by each agent, to a distribution over rewards for each of these agents in that period. We restrict our attention to the set of policies for which the above limit exists. More importantly, we impose that any policy $\varphi$ satisfy the following two fairness criteria:
\begin{enumerate}
    \item If two agents are in the system {in the same period}, the distribution from which their rewards are drawn is identical. Thus, $\varphi$ must map the set of  agents in the system in a given period, and the history of rewards observed by each agent, to a {\it single} distribution over rewards in every period. We denote the distribution in period $t$ by $\bx(t) = (x_r(t), r \in \rewardset)$. {Then, in a fixed period $t$, we have
        \begin{align*}
        \widehat{\Pi}(t) = \rev\left(\randomstate(t)\right)-\randomstate(t)\left(\sum_r rx_r(t)\right), 
        \end{align*}} 
\noindent \chedit{where the second term captures the expected payout across all $N(t)$ agents who each receive reward~$r$ with probability $x_r(t)$.}
    \item The {\it average reward distribution} observed by agents of different types, conditional on their types, is ``approximately'' the same over time. We refer to this latter constraint as {\it group fairness}, and provide its mathematical formalization in \cref{sec:main-results}.
\end{enumerate}

\chedit{At this point we reiterate that it is not our goal to prescribe our fairness definitions as the only reasonable ones across all industries. However,  given the spectrum of ways in which a decision-maker can differentiate between groups, it is unclear where one should draw the line. For instance, consider a policy in which type 1 agents receive a reward of 50 almost surely, whereas type 2 agents receive a reward of zero 49\% of the time, a reward of 100 49\% of the time, and a reward of 50 2\% of the time. The reward distributions seen by both agent types have the same expectation and median; however, depending on the context, one type may be perceived as having a more desirable reward distribution  (as Type 1 benefits from the certainty of a reward of 50 in each period). Our fairness definition may be particularly stringent in choosing where to draw the line, but it provides a simple first step to understanding how heterogeneous groups differentially self-selecting into the system may produce unfair outcomes.}
}

\subsection{Large-market regime}

{Given the size of the state space, the curse of dimensionality renders the goal of solving the MDP to optimality intractable. As a result, we turn to the more attainable goal of designing {\it asymptotically optimal} policies in a so-called {\it large-market limit}. 
}

The regime we consider is defined by a sequence of systems parametrized by $\theta \in \mathbb{N}^+$, with $\scaledlambda_i = \theta\lambda_i$ for all $i \in [K]$, and departure probabilities $\left\{\ell_i(\cdot)\right\}_{i\in[K]}$ held fixed. 
{We use~$\scaledrandomstate(t)$ to denote the number of agents in the system in period $t$ in the scaled system, and~$\scaledrandomstate$ the steady-state number of agents in the system. In order to keep the cost and revenue of any given policy on the same order, we define the 
{normalized profit in the scaled system} 
at time $t$ to be
\begin{align*}
\scaledprofit(t) = \rev\left(\frac{\scaledrandomstate(t)}{\theta}\right)-\frac{\scaledrandomstate(t)}{\theta}\left(\sum_r r x_r(t)\right).
\end{align*}
For $\theta \in \mathbb{N}^+$, we let $\scaledv(\varphi)$ denote the long-run average {normalized} 
profit of policy $\varphi$. 

{We briefly motivate this choice of scaling via a newsvendor-like revenue function. Suppose $\rev(N) = \min\{N,D\}$, for all $N \in \mathbb{N}$ and some $D > 0$. Then, scaling both demand and supply by $\theta$, 
a more standard scaling in which we simply divide $\rev$ by $\theta$ would give us:
\begin{align*}
    \frac1\theta\rev(N^\theta) = \frac1\theta \min\left\{N^\theta, \theta D\right\} = \min\left\{\frac{N^\theta}{\theta}, D\right\} = \rev\left(\frac{N^\theta}{\theta}\right).
\end{align*}
}
With strictly concave and increasing $R$, {a more standard scaling in which the revenue function is {not normalized} 
(i.e., $\rev{(N^\theta)}$ is not replaced with $\rev(N^\theta/\theta)$), yields a vacuous asymptotically optimal solution $\bx=\unitvec_{\rmin}$, and the objective going to $-\infty$. Scaling the system in this manner overcomes such vacuities.} Finally, we remark that a constant additive loss in $\widehat{\Pi}^\theta$ translates to an $\Omega(\theta)$ loss in the non-normalized profit $\theta\rev(\scaledrandomstate/\theta)-\scaledrandomstate(\sum_r rx_r)$.
}

\paragraph{Deterministic relaxation of the stochastic system.} 

In order to analyze {\it first-order} {differences between policies}, we consider a deterministic relaxation of the stochastic system, formally defined below.

Let $(\mathbf{x}(t))_{t \in \mathbb{N}^+}$ denote a sequence of reward distributions. In each period~$t$, for all~$i \in [K]$, we observe~$\lambda_i$ arrivals and~$\fluidn_i(t)\sum_r \ell_i(r)x_r(t)$ departures of type~$i$ agents, where $\fluidn_i(t)$ satisfies the following inductive relation:
\begin{align}\label{eq:inductive-deterministic}
    \fluidn_i(t+1) = \fluidn_i(t) + \lambda_i - \fluidn_i(t)\sum_r\ell_i(r)x_r(t), \quad \forall\, t \in \mathbb{N}^+.
\end{align}
Let $\widetilde{N}(t) = \sum_i \widetilde{N}_i(t)$, for $t \in \mathbb{N}^+$. Let $\widetilde{\varphi}$ denote the policy which pays out rewards from $(\mathbf{x}(t))_{t \in \mathbb{N}^+}$ in this deterministic system, and let $\detPi(\widetilde{\varphi})$ denote the long-run average profit induced by $\widetilde{\varphi}$ in the deterministic system, i.e.,
\begin{align}
    \detPi(\widetilde{\varphi}) = \lim_{T\to\infty}\frac1T \sum_{t=1}^T \left[R\left(\fluidn(t)\right) - \left(\sum_r r x_r(t)\right)\fluidn(t)\right].
\end{align}
Given policy $\varphi$, when necessary we use $\widetilde{N}^\varphi(t)$ to emphasize the dependence on $\varphi$. (As in the stochastic system, we assume $\widetilde{\varphi}$ is such that this limit exists.)

{In order to establish the connection between the two systems, consider the following coupling: for any $T, \theta \in \mathbb{N}^+$, fix payout distributions $\mathbf{x}(1),\ldots,\mathbf{x}(T)$. Moreover, let $\Delta(T)$ be the absolute difference between the expected profit in the stochastic system and the profit in the deterministic system over $T$ periods, i.e.,
$$\Delta(T) := \left\lvert\left(\sum_{t=1}^T\EE\left[R\left(\frac{N^\theta(t)}{\theta}\right)\right] -\sum_r rx_r(t)\EE\left[\frac{N^\theta(t)}{\theta}\right]\right)-\left(\sum_{t=1}^T R(\widetilde{N}(t))-\sum_r rx_r(t)\widetilde{N}(t)\right)\right\rvert.$$ The following proposition states that the long-run average profit in the deterministic system converges to the long-run average expected profit in the large-market regime.}

{
\begin{proposition}\label{prop:det-to-stoch}
Suppose $\widetilde{N}_i(0) = \mathbb{E}\left[\frac{N_i^\theta(0)}{\theta}\right]$ for all $i \in [K]$. Then, there exists a constant $C_0 > 0$ (independent of $T$, $\theta$) such that $\lim_{T\to\infty}\frac{\Delta(T)}{T} \leq C_0/\sqrt{\theta}$.
\end{proposition}
}

We conclude the section with a discussion of our modeling assumptions.

\subsection{Discussion of modeling assumptions}\label{ssec:modeling-asp} 

\paragraph{{Memoryless agents.}} \chedit{One limitation of our model is that agents decide to stay or leave based only on the most recent reward; in Appendix \ref{apx:memory-extension} we extend our fairness result to a setting in which agents have a two-period memory. The fact that the same results hold there suggests that similar results may hold for even more general models of memory. However,  generalizing the memoryless assumption of our model is not a main objective of our work.} Indeed, the assumption is motivated by the fact that {agents typically lack insight into the algorithms that generate the decisions they receive (e.g., why an algorithm paid out a reward in a given period, in our setting).}
{Moreover, the memoryless assumption follows a long tradition of models that consider agent attrition (also referred to as disengagement, in some settings) in the operations literature. For instance, within the context of recommender systems, \citet{ben2022modeling} and \citet{bastani2021learning} assume that the probability a user disengages with the recommendations depends only on the quality of the most recent recommendation. \citet{afeche2017customer} similarly note that this sort of  ``recency effect'' is typically assumed in the customer retention setting, in models that link demand to past {\it service levels} \citep{hall2000customer, ho2006incorporating, liu2007dynamic}. \citet{lemmens2020managing} also consider memoryless customers in their churn prediction problem.
{As a result, the focus of our work is not on improving 
 existing models of agent memory, but rather on leveraging existing memoryless}  models to gain insights into {\it fairness considerations} for these well-studied systems.}

We conclude the discussion of the memoryless assumption by noting that the abstraction of a period {is very general.} {For example, in a contractual employment setting, a period could be considered to be the length of the contract. In a noncontractual setting, a period could constitute however long agents are believed to consider past rewards before making the decision to leave the system. In the setting where agents are customers, a period would be the duration of the subscription contract.}

\paragraph{Time-invariance of the revenue function.} Another assumption upon which our model relies is the fact that the revenue function depends only on the number of agents in the system in a given period. {{In an employment setting, for instance,} the time-invariant assumption models a {\it mature} market, with newsvendor-like dynamics{. The work performed in the system can be viewed as ``low-skill,'' in the sense that workers arriving to the system are homogeneous, and the decision-maker does not benefit from workers gaining skill specificity with time. In the customer retention setting, on the other hand, stationarity of the revenue function is a reasonable assumption within the context of profit generated from the number of active subscribers {in a mature market} (ignoring heterogeneity in subscription plans).}} \chedit{An interesting question beyond the scope of our work is whether a dynamic policy outperforms a static policy when the revenue function is dictated by a state that can be either low or high; though a static policy would do much worse in such a setting, it is unclear whether a dynamic policy can outperform a static one without violating our fairness constraints.}

\paragraph{Unconditional versus conditional incentives.}  {As in \citet{lemmens2020managing}, we focus on {\it unconditional} incentives, wherein the decision-maker pays out the reward {\it independently} of the agent's decision to stay or leave, as opposed to {\it conditional} incentives. Empirical evidence of the effectiveness of such unconditional incentives has been found in the behavioral sciences, e.g., within the context of physician and patient surveys~\citep{abdulaziz2015national,young2015unconditional,rosoff2005response}, in addition to clinical study enrollment~\citep{young2020unconditional,kumar2022randomized}. {We believe that the analysis of conditional incentives can be similarly approached.}

	\section{Optimal Policies via the Deterministic Relaxation}\label{sec:main-results}

{In this section we design and analyze a heuristic policy within the context of the deterministic system.  We begin by formalizing the group fairness constraint, first introduced in \cref{sec:preliminaries}.
}

\begin{definition}[Group-fair policy]\label{def:fair_policy}
A policy $\widetilde{\varphi}$ defined by sequence of reward distributions $(\bx(t))_{t \in \mathbb{N}^+}$ is \emph{group-fair} if, for all $\delta > 0$, there exists $\tau_0 \in \mathbb{N}^+$ such that for all $\tau > \tau_0$:
\begin{align}\label{eq:fair-policy}
    \bigg{\lVert}\frac{1}{\sum_{t=t'}^{t'+\tau}\widetilde{N}^\varphi_i(t)}\sum_{t=t'}^{t'+\tau} \widetilde{N}^\varphi_i(t)\bx(t) - \frac{1}{\sum_{t=t'}^{t'+\tau}\widetilde{N}^\varphi_j(t)}\sum_{t=t'}^{t'+\tau} \widetilde{N}^\varphi_j(t)\bx(t) \bigg{\rVert}_1 < \delta  \qquad \, \forall \, t' \in \mathbb{N}^+, \, \forall \, i,j \in [K]. 
\end{align}
\end{definition}
Informally, a group-fair policy guarantees that, over any long enough time interval, the expected reward distributions respectively observed by different agent types do not differ too greatly.

We first show that, despite the unwieldiness of the group fairness constraint, there exists an exceedingly simple group-fair policy that is optimal in the context of the deterministic system: a policy that pays out the {\it same} distribution in each period.

\subsection{Optimality of the fluid-based heuristic}

 Consider the following optimization problem, termed \fluidopt, which computes the optimal {\it static} policy in the deterministic system described above:
\begin{align}\label{eq:fluid-opt}
    \optfluidprofit := \max_{\xvec \in \simplex^{|\rewardset|}, {\mathbf{\fluidn} \in \mathbb{N}^K}} &\rev\left(\sum_i \fluidn_i\right) - \left(\sum_r r x_r\right)\left(\sum_i \fluidn_i\right)\tag{\fluidopt}\\
    \text{s.t.} \qquad &\lambda_i=\fluidn_i\sum_r\ell_i(r)x_r \quad \forall \, i \in [K]. \notag 
\end{align}
Here, the stability constraint ensures that, for each type, the number of arrivals and departures are equal, and follows from plugging $x_r(t) = x_r$, for all $r \in \rewardset, t \in \mathbb{N}^+$ into \eqref{eq:inductive-deterministic}.  Note moreover that omitting the group fairness constraint \eqref{eq:fair-policy} is without loss of generality, as static policies are necessarily group-fair. We have the following theorem.

\begin{theorem}\label{thm:static-policies-are-opt-for-one-type}
Let ${\Phi}$ denote the space of all fair policies. Then, $\sup_{\varphi\in\Phi}\widetilde{\Pi}(\varphi) = \widetilde{\Pi}^*$. That is, there exists an optimal fair policy that is static.
\end{theorem} 

In the remainder of the paper, we refer to the optimal static policy as the \emph{fluid heuristic}. 
 \chedit{The proof of \cref{thm:static-policies-are-opt-for-one-type} is constructive. In particular, we show that the static policy which allocates each reward $r \in \Xi$ according to its long-run average probability under any fair dynamic policy induces a weakly higher long-run average revenue at a weakly lower cost, thus implying weakly improved profit. In fact, in Appendix \ref{apx:fairness-thm} we prove an even stronger statement: that in a system with {\it endogenous} arrivals, where types choose to join the system by comparing their respective long-run average rewards to a reservation wage (see Appendix \ref{apx:fairness-thm} for a formal specification of such a model), there exists an optimal fair policy that is static. Since exogenous arrivals are a special case of endogenous arrivals (i.e., all types have a reservation wage of zero), we obtain \cref{thm:static-policies-are-opt-for-one-type}.}

\subsection{Impact of discrimination by type}\label{sec:pof}

We next investigate the impact of the two fairness constraints imposed. In particular, when expanding the space of policies beyond fair ones, one approach a decision-maker could take would be in the flavor of \emph{learn, then discriminate}: by deploying machine learning algorithms to learn agents' types, a decision-maker can leverage this additional information to then pay agents of different types different amounts. We say that such policies {\it explicitly}  discriminate. 

{

\cref{prop:explicit-wage-disc} formalizes the intuition described above, that policies that learn agent types and target ``cheaper'' agents can greatly outperform optimal fair policies.

\begin{proposition}\label{prop:explicit-wage-disc}
Consider the setting with $K = 2$, $\rewardset = \{0,v_1,v_2\}$, $v_1 < v_2$, and the following departure probabilities:
\begin{align*}
    \ell_1(r') = \begin{cases}
    1 \quad &\mbox{if }r' = 0 \\
    0 \quad &\mbox{if }r' \in \{v_1,v_2\} 
    \end{cases} \qquad \text{and} \qquad 
    \ell_2(r') = \begin{cases}
    1 \quad &\mbox{if }r' \in \{0,v_1\}\\
    0 \quad &\mbox{if }r' = v_2.
    \end{cases}
\end{align*}
Moreover, let $R(\widetilde{N}) = \alpha \min\{\widetilde{N}, D\}$, $\alpha > 2v_2$, and {$\lambda_1 = D/4, \lambda_2 = D/2$}.
Then, there exists a policy~$\varphi^b$ that {\it explicitly}  discriminates such that $\widetilde{\Pi}(\varphi^b) - \widetilde{\Pi}(\varphi^s) = \Omega(D)$, where $\varphi^s$ is the optimal static policy. 
\end{proposition}

The policy $\varphi^b$ that we construct is {\it belief-based}, i.e., it targets cheaper type 1 agents by first learning their type, and then keeping them in the system, all the while keeping type 2 agents out of the system. Specifically, $\varphi^b$ learns the type of agents early on by paying all arriving agents~$v_1$. If an agent stays in the system after having been paid~$v_1$, then this agent is necessarily a type 1 agent, who is ``cheaper'' to keep in the system than a type 2 agent. Once enough type 1 agents are in the system, the policy no longer needs to keep arriving agents in the system, and can pay them nothing for the rest of time.

The above policy clearly violates our first fairness desideratum of drawing rewards from the same distribution for all agents {\it within} a given period. Our next result shows that there exist policies that satisfy this first fairness constraint, but fail to be group-fair; moreover, avoiding group-fairness allows this policy to outperform any fair policy by an unbounded amount. 
We refer to this more subtle version of discrimination, which pays agents in the same period according to the same distribution, as \emph{implicit  discrimination}. {In order to illustrate this, we introduce the notion of a {\it cyclic policy}.

\begin{definition}[Cyclic policy]\label{def:cyclic-pol}
Policy $\varphi$ is \emph{cyclic} if there exists $\tau \in \mathbb{N}^+$ such that $\bx(t+\tau) = \bx(t)$ for all $t \in \mathbb{N}^+$. The smallest $\tau$ for which this holds is the cycle length of policy $\varphi$, which we term $\tau$-\emph{cyclic}.
\end{definition}

When making the distinction between a $\tau$-cyclic policy $\varphi^\tau$ and another policy $\varphi$, we sometimes use $\widetilde{N}_i^{\tau}(t)$, for $i \in [K], t \in \mathbb{N}^+$. \cref{ex:steady-state-cyclic} shows that cyclic policies may implicitly discriminate, and outperform the optimal static policy.
}
}

\begin{proposition}\label{ex:steady-state-cyclic}
Suppose $K = 2$, and $\lambda_2 = \lambda, \lambda_1 = {0.1\lambda}$, $\lambda > 0$. Let $\Xi = \{0,r\}$, for some $r > 0$, with departure probabilities given by:
\begin{align*}
    \ell_1(r') = \begin{cases}
    0 \quad &\mbox{if }r' = r \\
    0.1 \quad &\mbox{if }r' = 0
    \end{cases} \qquad \text{and} \qquad 
    \ell_2(r') = \begin{cases}
    0.5 \quad &\mbox{if }r' = r \\
    1 \quad &\mbox{if }r' = 0.
    \end{cases}
\end{align*}
Suppose moreover that $\rev(\widetilde{N}) = \alpha \widetilde{N}$, $\alpha \in [0.7r, r)$. Consider the cyclic policy $\varphi^c$ of length $2$ which alternates between the two rewards in every period, i.e., the policy defined by $(\bx_r(t), t \in \mathbb{N}^+)$ such that:
\begin{align*}
    x_r(t) = \begin{cases}
    1 \quad &\mbox{if } t \text{ odd} \\
    0 \quad &\mbox{if } t \text{ even}.
    \end{cases}
\end{align*}
Then, $\widetilde{\Pi}(\varphi^c) -\widetilde{\Pi}(\varphi^s) = \Omega(\lambda).$
\end{proposition}

The cyclic policy described above engages in strategic {\it reward slashing}: it induces a large number of type 2 agents to stay in the system every other period, thus benefiting from their presence in the next period. In this next period, however, the decision-maker is able to retain all of its revenue as net profit by not incentivizing agents to stay in the system. 
\chedit{We show in \cref{prop:example-disc}} that under $\varphi^c$, while a type 1 agent in expectation receives the higher reward approximately 50\% of the time, a type 2 agent is only paid the higher reward 40\% of the time in expectation. Thus, $\varphi^c$ fails to satisfy the group-fairness constraint. {{(This instance technically violates the assumption that $\ell(\rmax) > 0$.} We chose 
{these inputs for ease of exposition; one can similarly construct instances where $\ell_1(\rmax) = \epsilon$ for small enough $\epsilon > 0$.})}

{In both examples constructed above, a reward of 0 and high exogenous arrival rates were chosen for clarity of exposition. One can similarly construct an example with $\rmin > 0$, and significantly smaller arrival rates (e.g., $\tilde{\lambda}_1 + \tilde{\lambda}_2 = 0.01(\lambda_1 + \lambda_2)$), with a significantly longer learning period / period of building up the number of agents in the system. Thus, these insights are not intrinsically tied to the exogenous arrival rate, or a large presence of ``free'' agents.} \chedit{Finally, this phenomenon is also not tied to the fact that agents join the system independent of their expected earnings. In this example, the expected reward of type 1 is $\frac{19}{39}\cdot r$, and the average reward of type 2 is $0.4r$  (see \cref{prop:example-disc} for a derivation). Hence, in the setting with endogeneous participation decisions (see Appendix \ref{apx:fairness-thm}), as long as type 1 and 2 agents have reservation wages of at most $\frac{19}{39}\cdot r$ and $0.4r$, respectively, both types would choose to join the system under an endogenous arrival model as in \cref{remark:endo}, and type 2 agents would still receive lower rewards on average.} 

{We conclude the section by noting that a natural quantity to consider is the {\it price of fairness}, i.e., the worst-case ratio (across all problem instances) between the decision-maker's optimal profit with all fairness constraints relaxed, and her profit under the optimal fair policy. A slight modification to the instance in \cref{ex:steady-state-cyclic} immediately gives us that the price of fairness in our setting is {\it unbounded}. To see this, let $v^*$ denote the value of the optimal fair solution for this instance. Defining revenue function $\widehat{R}(\widetilde{N}) = R(\widetilde{N}) - v^*$, the optimal fair solution for this new problem instance achieves a profit of zero, whereas the cyclic policy still achieves strictly positive profit, thus resulting in an unbounded price of fairness.}

\section{Structure and analysis of the fluid heuristic policy}\label{sec:asymptotic-opt}

Leveraging a deterministic relaxation to develop asymptotically optimal policies is a classical approach in many stochastic control problems. 
In these models, not only is the deterministic relaxation a natural upper bound, but it is also an attractive candidate to develop good policies due to its tractability; in particular, the corresponding fluid problem can typically be cast as either a linear or convex program. In this section, we first establish that our problem does not inherit such a convenient structure. {\cref{prop:fluid-non-cvx} formalizes this.}

{
\begin{proposition}\label{prop:fluid-non-cvx}
\ref{eq:fluid-opt} is, in general, non-convex.
\end{proposition}

{Thus, solving~\ref{eq:fluid-opt} efficiently is \emph{a priori} a difficult task given its high-dimensional and non-convex nature. Though the issue of non-convexity is typically circumvented via an exchange of variables~\citep{cao2020dynamic}, this approach fails in our setting due to the nonlinearity in the total cost of employing $\fluidn_i$ agents, $\fluidn_i\left(\sum_r rx_r\right)$, for $i \in [K]$. Further, the arbitrary heterogeneity in agent types (i.e., the lack of assumptions on $\left\{\ell_i(r)\right\}_{i \in [K]}$ stronger than non-increasing in $r$), makes a crisp characterization of the optimal solution seem elusive.}
} {Despite this, we derive structural properties that allow us to efficiently identify a solution to this a priori intractable problem. Thereafter, we show that our solution is asymptotically optimal for the stochastic system.}

\subsection{Structure and computation of the fluid optimum}\label{ssec:fluid-structure}
 
{\cref{thm:main-theorem} first formalizes that it is possible to derive tractable optimal solutions based on a surprising structural property.}

\begin{theorem}\label{thm:main-theorem}
For any instance of ~\ref{eq:fluid-opt}, independent of $K$ and $|\rewardset|$, there exists an optimal solution $\nf{\bx}^\star$ such that $|\supp(\nf{\bx}^\star)|\leq 2$.
\end{theorem}
{Theorem~\ref{thm:main-theorem} presents a {\it structural} characterization of the fluid optimal solution which has far-reaching implications for the {\it analytical} and {\it computational} tractability of~\ref{eq:fluid-opt}.  In particular, suppose we fix $r_1,r_2 \in \rewardset$ and assume $\supp(\bx^\star) = \{r_1, r_2\}$, with $x$ being the weight placed on $r_1$, and $1-x$ the weight placed on $r_2$. Then, \ref{eq:fluid-opt} becomes a {one-dimensional} problem in $x$, and can be solved via the KKT conditions. This then implies that 
an {\it optimal solution to the non-convex problem} \ref{eq:fluid-opt} can be found efficiently {by exhaustively searching over $\binom{|\rewardset|}{2}$ possible pairs of rewards}.} 
We provide some high-level intuition for the strategy used to prove this key structural result below, and defer the proof of the theorem to Appendix \ref{apx:structure-of-fluid}.

Our proof technique is motivated by the following natural interpretation: a profit-maximizing solution must trade off between recruiting enough agents in order to collect high revenue, all the while keeping costs relatively low. We disentangle these two competing effects by introducing a closely related {\it budgeted supply maximization} problem, which we term~\ref{eq:supply-opt}. {At a high level,~\ref{eq:supply-opt} simplifies~\ref{eq:fluid-opt} by removing a degree of freedom: namely, {\it how much the decision-maker can spend to retain agents}. Given this added constraint, there are no longer two competing effects: the goal of the decision-maker is to simply recruit as many agents as possible.} We then constructively show that, given an optimal solution to \ref{eq:supply-opt} that places positive probability mass on more than two rewards, there exists a ``support reduction'' procedure that takes three appropriately chosen rewards and allocates all of the weight on one reward to the other two. Crucially, this procedure is designed so that both the average reward paid out to agents and the total number of agents are maintained. {We complete the proof by showing that, if an optimal solution to \ref{eq:supply-opt} has a support of size no more than 2, it must be that the same holds for \ref{eq:fluid-opt}.}

\subsection{Asymptotic optimality of the fluid heuristic in the stochastic system}

We now analyze the performance of the fluid-based heuristic in our original stochastic system. Let $\varphi^\star$ be the policy that draws rewards according to $\bx^\star$ in every period. 
\cref{prop:fluid-ub-for-all-theta} first establishes that $\optfluidprofit$ is an upper bound on the profit of the optimal policy in the system parameterized by~$\theta$.

\begin{proposition}\label{prop:fluid-ub-for-all-theta}
Let $\scaledv^\star~=~ \sup_{\varphi\in\Phi}\scaledv(\varphi)$. $v^\star_\theta \leq \optfluidprofit$ for all $\theta \in \mathbb{N}^+$.
\end{proposition}

{
As a corollary of \cref{prop:det-to-stoch} and \cref{prop:fluid-ub-for-all-theta}, we obtain the standard $\mathcal{O}\left(\frac{1}{\sqrt{\theta}}\right)$ additive {loss} bound of the fluid heuristic. {Under a mild additional technical condition, we further prove that~$\scaledv(\varphi^\star)$ actually converges to this upper bound at a {\it linear} rate.}
Let {$\Lambdaub = \sum_i \frac{\lambda_i}{\ell_i(\rmax)}$ and $\Lambdalb = \sum_i \frac{\lambda_i}{\ell_i(r_{\min})}$}. Theorem~\ref{thm:asymptotic-result} characterizes the performance of policy $\varphi^\star$ relative to $\optfluidprofit$.}

\begin{theorem}\label{thm:asymptotic-result}
Suppose $\rev$ is twice-continuously differentiable over $\mathbb{R}_{> 0}$, and that there exists a constant $\alpha > 0$ such that $\rev''(n) \geq -\theta^{\alpha}$ for all $n \in \left[\frac1\theta,\Lambdaub\right)$. Then, there exists a constant $C > 0$ such that  $\scaledv(\varphi^\star) \geq {\optfluidprofit} - \frac{C}{\theta}.$
\end{theorem}

It is easy to check that ``reasonable'' concave functions such as $\rev(n) = n^{\beta}$, $\beta \in (0,1)$, and $\rev(n) = \log(1+n)$ satisfy the conditions of Theorem~\ref{thm:asymptotic-result}.

At a high level, our system distinguishes itself from {many} systems considered in the classical stochastic control literature due to the fact that it is ``self-adjusting,'' or, more specifically, ``self-draining.'' What we mean by this is that, even though $\varphi^\star$ is a static policy, its effect on the system does in fact depend on the current state: the higher the number of agents in the system, the more agents depart from the system in the next period, on average. This {directly} follows from the fact that the mean of a Binomial distribution is increasing in the number of trials. Such a ``self-draining'' feature acts as self-regulation, preventing the system from being overloaded with agents who yield little marginal revenue relative to the reward paid out. On the other end of the spectrum, the lower bound on $\rev''$ for small values of $x$ precludes exponentially steep functions, for which the revenue loss could be large when there are few agents in the system. A similar ``self-adjustment'' phenomenon leading to $\mathcal{O}\left(\frac1\theta\right)$-convergence was identified in~\citet{cao2020dynamic}, though for an entirely different setting.

\paragraph{Numerical experiments.} We conclude this section by complementing our theoretical results with numerical experiments. In particular, we illustrate the  performance of $\varphi^\star$ relative to two natural reward schemes a decision-maker may use: a deterministic payout in each period, and a lottery. We first describe the experimental setup.

\medskip

\emph{{Reward} schemes.} We assume the revenue function is newsvendor-like: $\rev(x) = 100\min\left\{x, 150\right\}$ for $x \in \mathbb{R}_{+}$. We let $\rewardset = \{15, {16,}\ldots,60\}$ and consider three different reward schemes: $(i)$ 
the fluid-based heuristic, $(ii)$ {a lottery with variance $\widehat{\sigma}^2$ such that agents receive $\mu$ in expectation, with $\mu = \sum_r rx_r^\star$ and $\widehat{\sigma} = 10$,} and $(iii)$ a fixed reward $r^{\text{det}}$ (specifically, the optimal fixed reward).

\medskip

\emph{Agent types.} We let $K = 3$, with $\lambda_1 = \lambda_2 = \lambda_3 = 10/3$, and consider the following departure probability functions, depicted in Figure~\ref{fig:fluid-probs}:
\begin{enumerate}
    \item type 1 agents: $\ell_1(r) = \min\left\{1,e^{\alpha_1(-r+15)}\right\}$, $\alpha_1 = 7/100$
    \item type 2 agents: $\ell_2(r) = -\alpha_2r + \beta_2$, $\alpha_2 = 1/45$, $\beta_2 = 4/3$
    \item type 3 agents: $\ell_3(r) = -\alpha_3r^2 + \beta_3r + \gamma_3$, $\alpha_3 = (1/2025), \beta_3 = 2/135, \gamma_3=8/9$.
\end{enumerate}

\medskip

\emph{Performance metric.} 
For each of the three reward schemes $\varphi$ described above, we compute
$\optfluidprofit - \scaledv(\varphi)$, for $\theta\in\{1,\ldots,5\cdot 10^3\}.$

\medskip

\emph{Results.} We show the results {(for 1,000 replications of simulations)} in Figure~\ref{fig:fluid-regret}. These results illustrate the strong performance of the fluid heuristic for this particular problem instance. In particular, we indeed observe that the fluid heuristic converges to $\optfluidprofit$, 
as predicted by our theoretical results. Moreover, the fluid policy 
outperforms the deterministic and lottery schemes, {both of which have performance loss bounded away from zero, i.e., they are not (asymptotically) optimal.} \chedit{We refer the reader to Appendix \ref{app:small_market_numerics} for further investigations of the asymptotic convergence, including in moderately sized markets. }

\begin{figure}
	\captionsetup[subfigure]{justification=centering}
	\centering
	\subfloat[$\ell_i(r)$ vs. $r$, for $i \in \{1,2,3\}$ \label{fig:fluid-probs}]{{\includegraphics[width=0.45\textwidth]{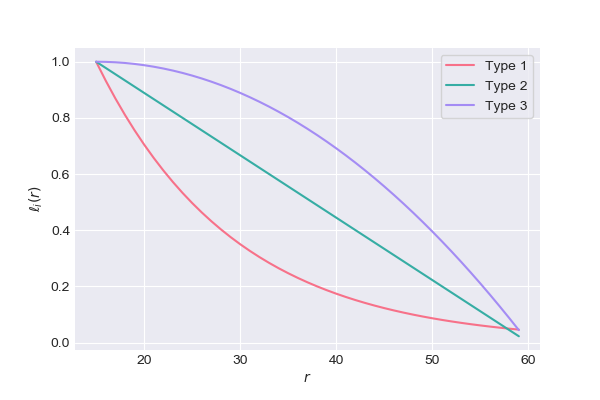} }}\quad%
	\subfloat[Additive loss vs. $\theta$\label{fig:fluid-regret}]{{\includegraphics[width=0.45\textwidth]{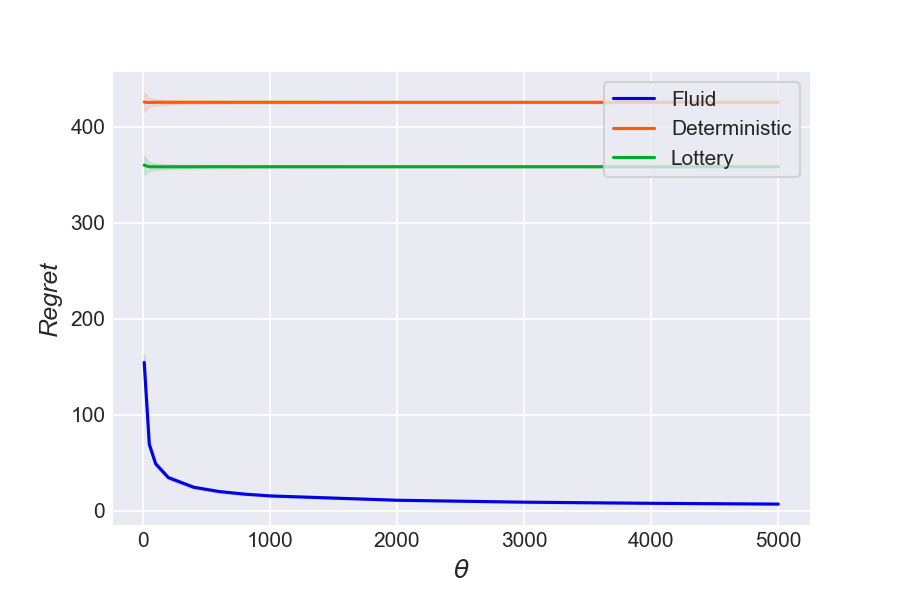} }}
	\caption{Additive loss of the fluid (dark blue curve), deterministic (dark orange curve), and lottery (green curve) heuristics.}%
	\label{fig:ogd-regret}%
\end{figure}


\section{Special cases}\label{sec:special-cases}

\chedit{The tractable structure of the fluid optimum enables not only our numerical experiments, but also allows us to gain additional analytical insights into optimal policies under more structured departure probability models. We first investigate how the optimal policy depends on the {\it convexity} of the departure probability function, and subsequently consider the optimal policy when agents' departure probability function is a type of softmax function.  
}

{
\subsection{Impact of convexity of departure function}\label{sec:cvx-ccv}

To understand the impact of the structure of the agents' departure probability functions on the optimal reward scheme, we define the notions of {\it convexity} (resp., {\it concavity}) over the reward set, as well as the {\it dispersion} of a reward scheme.
\begin{definition}\label{def:ccv}
The departure probability function of a type $i$ agent is \emph{strictly convex} if:
$$\ell_i(r_2) < \ell_i(r_1)\cdot \frac{r_2-r_3}{r_1-r_3} + \ell_i(r_3)\left(1-\frac{r_2-r_3}{r_1-r_3} \right) \qquad  \forall \, r_1 > r_2 > r_3 \in \rewardset.$$

The departure probability function of a type $i$ agent is \emph{strictly concave} if:
$$ \ell_i(r_2) > \ell_i(r_1)\cdot \frac{r_2-r_3}{r_1-r_3} + \ell_i(r_3)\left(1-\frac{r_2-r_3}{r_1-r_3} \right) \qquad \forall \, r_1 > r_2 > r_3 \in \rewardset.$$
\end{definition}

\revedit{
\begin{remark}
At a high level, concave departure probabilities correspond to ``risk-seeking'' agents, since for these agents the expected departure probability from any convex combination of two rewards is higher than the departure probability from handing out the convex combination of these two rewards deterministically. An example of this would be people playing the lottery with minuscule odds of claiming the jackpot and (due to the cost of a ticket) a negative return in expectation. 
Conversely, convex types can be viewed as ``loss averse,'' as they are more likely to stay if paid out a reward deterministically. This  behavior may be more prevalent when there is a significant cost to remaining active in a rewards program. For instance, the Chase Sapphire Reserve card costs \$550 a year, and comes with a deterministic \$300 travel credit (and a range of other benefits), that is renewed each year. There, we would assume that most customers prefer the certainty of their travel benefits over  a small chance to win a more valuable travel credit.
\end{remark}
}

We next introduce the notion of {\it dispersion} of a reward scheme. 

\begin{definition}[Maximal and minimal dispersion]
Consider any feasible solution $\bx$ to~\ref{eq:fluid-opt}, and suppose $\supp(\bx) = \left\{r, r'\right\}$, with $r > r'$. $\bx$ is said to be \emph{maximally dispersed} if $r = \rmax$ and $r' = \rmin$. $\bx$ is said to be \emph{minimally dispersed} if $r$ and $r'$ are consecutive rewards in $\rewardset$, or if $\bx = \unitvec_r$ for some $r \in \rewardset$.
\end{definition}

{The following proposition formalizes the intuitive idea that the dispersion of an optimal policy vastly differs depending on {the convexity of $\ell_i(\cdot)$}. }

\begin{proposition}\label{thm:main-theorem2}
 If $\ell_i(\cdot)$ is strictly convex for all $i \in [K]$, then $\nf{\bx}^\star$ is \emph{minimally} dispersed. On the other hand, $\ell_i(\cdot)$ is strictly concave for all $i\in[K]$, then $\bx^\star$ is \emph{maximally} dispersed.
\end{proposition}

\cref{thm:main-theorem2} tells us that, if $\ell_i(\cdot)$ is strictly convex for all types, in the limit where $\rewardset$ is a compact convex set, it is optimal for the decision-maker to guarantee a fixed reward; {in contrast, for types with strictly concave $\ell_i(\cdot)$}, the decision-maker should run a lottery. 
{The proof of \cref{thm:main-theorem2} crucially relies on the two-reward structure of the optimal static solution derived in Section \ref{ssec:fluid-structure}, thus highlighting its importance in terms of deriving insights into optimal compensation schemes for practical special cases of agent behavior.} 

We next numerically investigate whether these analytical insights when agents have homogeneous structure of departure probability functions port over to mixed populations.}

{
\paragraph{Numerical insights.\label{ssec:risk}}
We consider a setting with two types of agents defined by the following departure probability functions: 
\begin{align*}
\ell_1(r) = e^{-a_1r}, \ a_1 \in\{0.5,1,2,3,4\} \qquad \text{and} \qquad \ell_2(r) = \frac{1-e^{2(r-1)}}{1-e^{-2}}.
\end{align*}
\cref{fig:cvx-ccv-probs} illustrates these departure probabilities.

\begin{figure}
     \centering
     \subfloat[$\ell_1(r)$ vs. $r$]{\includegraphics[width = 0.33\textwidth]{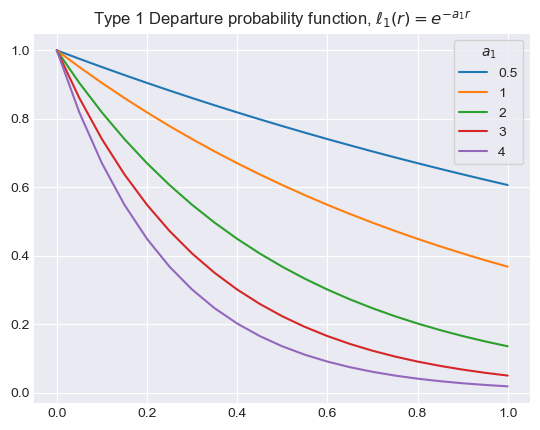}} \qquad
          \subfloat[$\ell_2(r)$ vs. $r$]{\includegraphics[width = 0.33\textwidth]{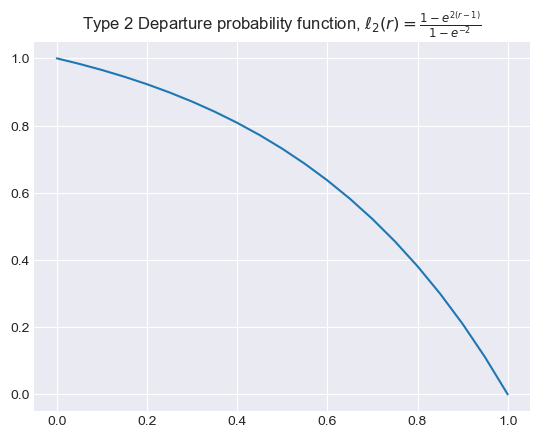}}
    \caption{Departure probability functions for convex and concave mixtures}
    \label{fig:cvx-ccv-probs} 
\end{figure}

The total arrival rate is normalized to 1, with $\alpha$ fraction of arrivals in each period being type 1 (convex) agents, and $1-\alpha$ being type 2 (convex) agents. 
The revenue function is newsvendor-like, with $R(N) = 5\min\{N,5\}$. Finally, $\Xi = \{0,0.05,\ldots,1\}$. 

\cref{fig:cvx-ccv-results} shows the coefficient of variation of the optimal reward distribution for $\alpha \in [0,1]$. Observe that the highest coefficient of variation occurs when $\alpha = 0$ (i.e., all arrivals have concave departure probability function), in line with our theoretical result around maximal dispersion. Conversely, the lowest coefficient of variation occurs at $\alpha = 1$ (i.e., all arrivals have convex departure probability function). In between, the coefficient of variation is decreasing in the fraction of concave arrivals (modulo slight perturbations due to discretization). For $a_1 \in \{2,3,4\}$ we observe a sharp phase transition: there exists a threshold fraction of convex arrivals past which the coefficient of variation of the optimal reward scheme is close to zero. As the convexity of $\ell_1$ increases, this threshold decreases. 

\begin{figure}[h!]
     \centering{\includegraphics[scale=0.35]{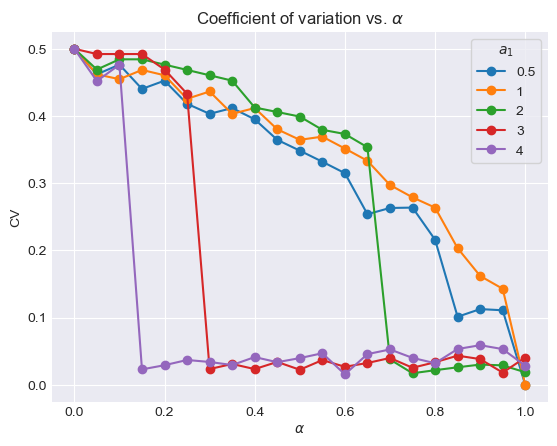}}
    \caption{Impact of population composition on dispersion of optimal reward distribution}
    \label{fig:cvx-ccv-results}
\end{figure}
}


{
\subsection{Linearized S-shape departure functions}

While \cref{sec:cvx-ccv} considered the setting where types' departure probability functions were either convex or concave, in this section we study the {setting in which} agents' departure probability functions that is similar to a softmax function. We define the notion of a ``linearized S-shape'' departure probability function below.
{
\begin{definition}[Linearized S-shape]\label{def:eps-noisy}
Departure probability function $\ell_i(\cdot)$ has a linearized S-shape if there exists $\epsilon, v_i > 0$ such that
\begin{align*}
    \ell_i(r) = \begin{cases}
    1 \quad &r < v_i-\epsilon \\
    \frac12-\frac{r-v_i}{2\epsilon} \quad &r \in [v_i-\epsilon, v_i + \epsilon] \\
    0 \quad &r > v_i+\epsilon
    \end{cases}
\end{align*}
\end{definition}
}

\begin{remark}
{In \cref{def:eps-noisy} we have relaxed the assumption that $\ell_i(\rmax) > 0$ for all $i \in [K]$. This however is without loss of generality for the class of revenue functions we consider in this section, as for these functions the profit would go to $-\infty$ if the policy were to deterministically pay out any reward $r \geq v_i+\epsilon$.}
\end{remark}
Figure \ref{fig:eps-noisy-loss-probs} illustrates these departure probabilities, for a fixed $v_i = 25$ and various values of~$\epsilon$. At a high level, $\epsilon$ can be seen as controlling agents' tolerance of reward uncertainty. As $\epsilon$ approaches 0, the above departure probability function mirrors that of an agent who, {upon receiving reward~$r$}, deterministically stays or leaves. As $\epsilon$ goes to $\infty$, we obtain $\ell_i(r) \to \frac12$ for all $r \in \Xi$, {and} the decision to stay or leave {becomes} independent of the reward paid out; {phrased differently, agents act independently of the reward they receive}.

\begin{figure}
     \centering{\includegraphics[scale=0.5]{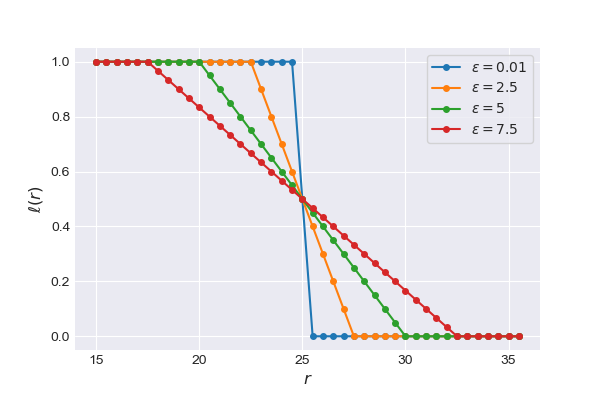}}
    \caption{Departure probability function $\ell(r)$ for various values of $\epsilon$, with $K = 1$ and $v = 25$.}
    \label{fig:eps-noisy-loss-probs}
\end{figure}

In the remainder of the section, for analytical tractability our results pertain to the setting where $K = 1$. However, our structural result regarding the fluid optimum allows us to gain numerical insights into {settings with $K > 1$}, as we will see later on in the section. For ease of notation, we omit the dependence of all quantities on the agent type. We moreover assume that $\rewardset$ is such that $v \in (\rmin, \rmax)$. 
Finally, we make the following assumption on the revenue function. 

\begin{assumption}\label{asp:c2} $\rev$ is twice-continuously differentiable, {with $\rev''(x) < 0$ for all $x \in \mathbb{R}_+$.} 
\end{assumption}

The following proposition characterizes the optimal solution $\bx^\star(\epsilon)$ to \ref{eq:fluid-opt}.

\begin{proposition}\label{prop:noisy-opt}
{Suppose $\epsilon \leq \min\{v-\rmin,\rmax-v\}$}. The optimal reward distribution $\bx^\star(\epsilon)$ to \ref{eq:fluid-opt} is a lottery that randomizes between $\rmin$ and $v+\epsilon$. Let $x^\star(\epsilon)$ denote the weight that $\bx^\star(\epsilon)$ places on $v+\epsilon$. Then, $x^\star(\epsilon) < 1$ for all $\epsilon > 0$. Moreover, the following holds:
\begin{enumerate}
    \item If $\rev'(\lambda) \leq v$, then $x^\star(\epsilon) = 0$ for all $\epsilon > 0$.
    \item If $\rev'(\lambda) > v$, then:
    \begin{enumerate}
    \item for all $\epsilon \leq \rev'(\lambda)-v$, $x^\star(\epsilon)$ is continuously decreasing in $\epsilon$, and
    \item for all $\epsilon > \rev'(\lambda)-v$, $x^\star(\epsilon) = 0$.
    \end{enumerate}
\end{enumerate}
\end{proposition}

We provide some intuition for the above result. $\rev'(\lambda)-v$ can be interpreted as the marginal revenue the decision-maker obtains for an agent, in addition to the $\lambda$ arrivals per period, net of the opportunity cost of that additional agent{ as $\epsilon$ approaches 0}. Thus, when $\rev'(\lambda)-v\leq 0$, agents are too costly to keep in the system. On the other hand, when $\rev'(\lambda)-v > 0$, it is worthwhile to incentivize agents to stay on with some probability. However, this benefit decreases as agents make noisier decisions, and as a result become more and more costly.

{In the remainder of the section, we let $\epsilon_0 = \rev'(\lambda)-v$ be the threshold past which it becomes too costly to keep agents in the system. We make the following mild assumption relating $\epsilon_0$, agents' value $v$, and the reward set $\rewardset$.

\begin{assumption}\label{asp:non-trivial}
$\epsilon_0 \leq \min\{v-\rmin, \rmax-v\}$.
\end{assumption}

{On the one hand, $\epsilon_0 \leq \rmax-v$ is a weak assumption that enforces that the maximum reward paid out by the decision-maker is at least the marginal revenue from the $(\lambda+1)$st agent in the system. The condition that $\epsilon_0 \leq \min\{v-\rmin\}$ can be viewed as a non-triviality condition: it simply imposes that $\epsilon_0$ not be so large that the decision-maker finds it optimal to pay out the minimum possible reward with probability 1 and use the newly arriving agents ``for free''. 

{Let $\widehat{\Pi}(\bx^\star(\epsilon);\epsilon)$ denote the decision-maker's optimal profit as a function of $\epsilon$.} We have the following fact.}
}

\begin{proposition}\label{prop:noisy-profit}
 $\widetilde{\Pi}(\bx^\star(\epsilon);\epsilon)$ is non-increasing in $\epsilon$.
\end{proposition}

At a high level, one would expect the decision-maker to be able to leverage the fact that agents are less sensitive to low rewards as $\epsilon$ increases, and can spend less on average as a result. \cref{prop:noisy-profit} invalidates this false intuition: as $\epsilon$ increases, the fact that agents must be rewarded more frequently to {\it guarantee they stay} in the system counterbalances these potential gains, and as a result, profit cannot increase.

Figure \ref{fig:noisy-metrics-many-types} shows that these trends persist even in the presence of many types, for a newsvendor-like revenue function. In particular, in all settings considered, we observe that the profit exhibits a threshold structure, in which there exists $\epsilon_0$ such that it decreases for $\epsilon \leq \epsilon_0$, and remains constant for $\epsilon > \epsilon_0$. The threshold $\epsilon_0$ at which this occurs, however, depends on the specific arrival rates; the threshold {\it decreases} as the arrival rate of agents with higher values grows. This is intuitive as these agents are the most ``expensive'' to keep in the system.
The extent to which this threshold depends on the type composition of arrivals, however, becomes more muted as $D$ grows large. This is because, for large enough $D$, there aren't enough cheap agents to satisfy all demand. As a result, the decision-maker must incentivize more expensive agents to stay on.

\begin{figure}[h!]
     \centering
     \subfloat[Profit vs. $\epsilon$, $D = 25$
     \label{fig:normal-variance-type1}]{\includegraphics[width = 0.33\textwidth]{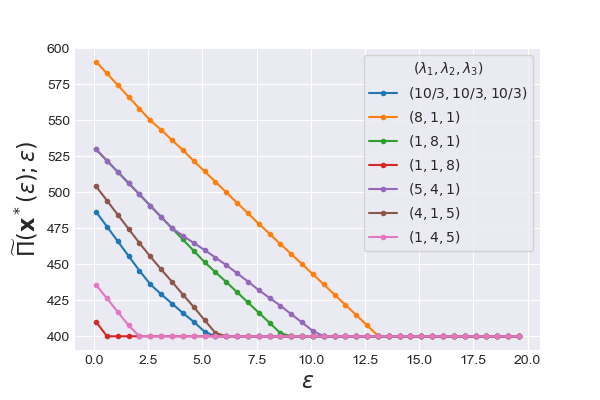}}
          \subfloat[Profit vs. $\epsilon$, $D = 100$
     \label{fig:normal-variance-type2}]{\includegraphics[width = 0.33\textwidth]{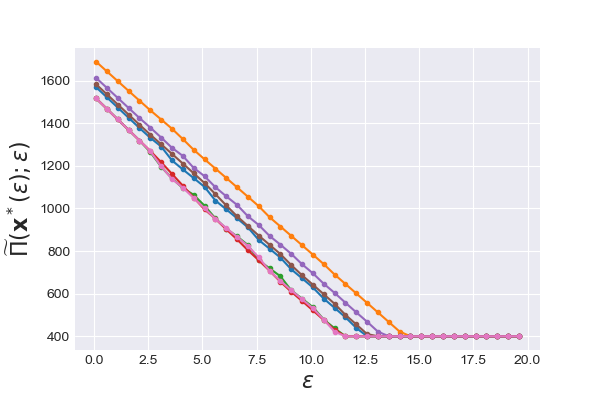}}
          \subfloat[Profit vs. $\epsilon$, $D = 400$
     \label{fig:normal-variance-type3}]{\includegraphics[width = 0.33\textwidth]{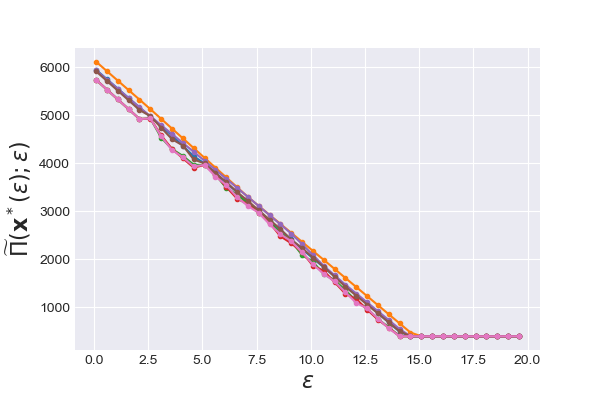}}
    \caption{Dependence of profit on $\epsilon$ for three types, with $\lambda_1 + \lambda_2 + \lambda_3 = 10$, $v_1 = 25, v_2 = 30, v_3 = 40$, and $\rev(N) = 40\min\{N,D\}$, with $D \in \{25,100,400\}$.}
    \label{fig:noisy-metrics-many-types} 
\end{figure}
}

	\section{Conclusion}\label{sec:discussion}

\chedit{
Our paper studies a decision-maker aiming to design {\it fair} incentives schemes when agents make stochastic participation decisions based on recent rewards. {Fairness constraints in this setting} lead to {an a priori non-obvious} obstacle for {dynamic reward} policies: when different agent types {exhibit} heterogeneous reactions to different rewards, 
dynamic policies can induce some types to be overrepresented in periods with lower rewards. Essentially, this is due to agents of these types self-selecting into those low-rewards periods based on the preceding high-reward periods.  
In the long run, this can result in bias in the reward distributions they experience. 
When dynamic policies are restricted to avoid such bias, we find that they offer no asymptotic benefit compared to static ones. Moreover, we prove, under a weak technical condition on the decision-maker's revenue function, that the asymptotic benefits of dynamic policies vanish fast, as the static policy converges to the fluid upper bound at a linear rate. Finally, leveraging the two-reward structure of optimal solutions to the fluid upper bound, {we} derive insights into the type of policies that perform best for {certain special cases of} departure probability functions.
}

    \bibliographystyle{apalike}
    {\bibliography{main}}
    
    \newpage
    \onecolumn
    \appendix
    
\chedit{
\section{Motivating application: Credit Karma's Instant Karma rewards program}

In this section we describe the applicability of our model to the Instant Karma random rewards program, a feature launched by Credit Karma Money (a checking account service operated by Credit Karma) in 2021 \citep{creditkarma}. Under this program, an instant refund may randomly occur anytime a Credit Karma Money customer uses their debit card at {selected stores.} 
Hence, customers effectively play a lottery every time they use their debit card.

We use our model to approximate the Instant Karma program via the following mapping:
\begin{itemize}
\item A customer being ``in the system'': the customer is an active user of the Credit Karma debit card; 
\item Length of a period: one year (e.g., the end of the period corresponds to the time at which a customer re-evaluates their card usage). At the end of year $t$, the user decides whether to stay based on the rewards received in year $t$. If so, she remains an active user throughout year $t+1$; if not, she becomes inactive at the beginning of year $t+1$.
\item Reward received at the end of the period: total refunds received throughout the year;
\item Decision-maker's revenue: in-app advertising \citep{creditkarma_rev}.
\end{itemize}

We now comment on how well our modeling assumptions fit the Instant Karma application:
\begin{itemize}
\item Exogenous arrivals: Though the possibility of refunds may influence the likelihood that individuals become active users of the Credit Karma debit card, the details of the program, such as the probability of receiving a refund, are not widely advertised by the company. Hence, it seems unlikely that potential new customers would form a prior over expected future refunds, and make their decision to become active users based on this. {However, we highlight that our main result on asymptotic optimality of static policies continues to hold under standard models of endogenous entry.}
\item Memorylessness: Given that we assume customers re-evaluate their debit card usage yearly, we expect that customers tend to have heavy recency bias in this setting.
\item Type-independent revenue function: This would hold under the assumption that Credit Karma curates in-app advertisements such that different customer types have similar click probabilities for ads with similar per-click revenues.
\item Time-invariant revenue function: This would hold under the assumption that both the kinds of ads shown to customers don't change and their click probabilities are time-invariant.
\item Unconditional rewards: This holds since refunds are immediately posted to the account, independent of the customer's decision to remain an active user at the end of year $t$.
\item {Large-market regime: In the year that Credit Karma launched this program, they reported refunding 100,000 transactions \citep{creditkarma}. Assuming that only a small fraction of transactions were refunded, and given Credit Karma's large customer base, it is reasonable to assume that it is operating in a large-market regime.}
\end{itemize}
}

\section{Connection between the stochastic and deterministic systems: Proofs}

\subsection{Proof of \cref{prop:det-to-stoch}}
\begin{proof}
We have:
\begin{align}
    \Delta(T)
    &\leq\sum_{t=1}^T\left\lvert\EE\left[R\left(\frac{N^\theta(t)}{\theta}\right)-R(\widetilde{N}(t))\right]\right\rvert +  r_{\max}\sum_{t=1}^T\left\lvert\EE\left[\frac{N^\theta(t)}{\theta} - \widetilde{N}(t)\right]\right\rvert \notag \\
    &\leq (L+r_{\max})\sum_{t=1}^T \left\lvert\EE\left[\frac{N^\theta(t)}{\theta} - \widetilde{N}(t)\right]\right\rvert \label{eq:lip}\\
    &= (L+r_{\max})\sum_{t=1}^T\sqrt{\left(\EE\left[\frac{N^\theta(t)}{\theta} - \widetilde{N}(t)\right]\right)^2} \notag \\
    &\leq (L+r_{\max})\sum_{t=1}^T\sqrt{\EE\left[\left(\frac{N^\theta(t)}{\theta} - \widetilde{N}(t)\right)^2\right]} \label{eq:jen}.
\end{align}
where \eqref{eq:lip} follows from $L$-Lipschitz continuity of $R$, and \eqref{eq:jen} from Jensen's inequality. The following lemma establishes the first connection between the number of agents in the deterministic systems and the expected number of agents in the large-market regime. We defer its proof to the end of the section.

\begin{lemma}\label{lem:det-stoch-num}
For all $i \in [K], \theta \in \mathbb{N}^+, t \in \mathbb{N}^+$, $N_i^\theta(t) \sim \text{Poi}(\theta\widetilde{N}_i(t))$.
\end{lemma}

Using the above fact, we then have that
\begin{align*}
\EE\left[\left(\frac{N^\theta(t)}{\theta} - \widetilde{N}(t)\right)^2\right] &= \text{Var}\left(\frac{N^\theta(t)}{\theta}\right) 
= \frac{1}{\theta^2}\sum_{i\in[K]}\EE[N_i^\theta(t)]
= \frac{1}{\theta}\sum_{i\in[K]}\widetilde{N}_i(t),
\end{align*}
where the second equality follows from the fact that the variance of a Poisson distribution is equal to its mean.
Since $\ell_i(\rmax) > 0$ for all $i$, there exists $N_{\max}$ such that $\widetilde{N}_i(t) \leq N_{\max}$ for all $i \in [K], t \in \mathbb{N}^+$. Plugging this back into \eqref{eq:jen}, we have:
$
    \Delta(T) \leq (L+\rmax)T\sqrt{\frac1\theta K N_{\max}}.
$
Dividing by $T$ on both sides, we obtain the result.\hfill\Halmos
\end{proof}

\begin{proof}[Proof of \cref{lem:det-stoch-num}.]
By Poisson thinning, $N_i^\theta(t)$ is Poisson distributed for all $t\in\mathbb{N}^+$, with $\EE[N_i^\theta(t)] = \EE[N_i^\theta(t-1)(1-\sum_r\ell_i(r)x_r(t))] + \theta\lambda_i$. Dividing by $\theta$ on both sides:
\begin{align*}
   \EE\left[\frac{N_i^\theta(t)}{\theta}\right] &= \EE\left[\frac{N_i^\theta(t-1)}{\theta}\right](1-\sum_r\ell_i(r)x_r(t)) + \lambda_i.
\end{align*}
Recall, $\widetilde{N}_i(t) = \widetilde{N}_i(t-1)(1-\sum_r \ell_i(r)x_r(t)) + \lambda_i$ for all $i$. Initializing $\widetilde{N}_i(0) = \mathbb{E}\left[\frac{N_i^\theta(0)}{\theta}\right]$, we obtain the result.\hfill\Halmos
\end{proof}



\section{Optimal Policies via the Deterministic Relaxation}

\subsection{Proof of \cref{thm:static-policies-are-opt-for-one-type}}\label{apx:fairness-thm}

\chedit{

{To prove the result, we consider an {\it endogenous} model of behavior wherein customers strategically make the decision to join the system based on whether their expectation of average rewards exceeds a given, nonnegative reservation wage. Noting that exogenous arrivals are a special case of this setting by taking all reservation wages to be zero, we obtain the result.} 

To formalize this, consider a policy $\widetilde{\varphi}$ defined by sequence of reward distributions $(\bx(t))_{t \in \mathbb{N}^+}$, and let $\widetilde{N}_i^{\varphi}(t)$ refers to the {\it counterfactual} number of type $i$ agents that would exist in the deterministic system in period $t$ if they did choose to join in each period, {as specified in \eqref{eq:inductive-deterministic}}. We assume type $i$ agents are endowed with a reservation wage $V_i \geq 0$, and join the system if and only if the long-run average reward of type $i$ agents exceeds that reservation wage, i.e.,:
\begin{align}\label{eq:joining-condition}
\lim_{T\to \infty} \sum_{t=1}^T \frac{\widetilde{N}_i^{\varphi}(t)\sum_r rx_r(t)}{\sum_{t=1}^T \widetilde{N}_i^{\varphi}(t)} \geq V_i.
\end{align}
We restrict our attention to the set of policies for which the limit on the left-hand side exists, and adapt our notion of group fairness to apply to all types who choose to join the system.

\begin{definition}[Endogenous group-fair policy]\label{def:endo_fair_policy}
For any policy $\varphi$, let $\mathcal{I}^{\varphi} \subseteq [K]$ denote the set of types for whom \eqref{eq:joining-condition} holds under $\varphi$. $\varphi$ is \emph{group-fair} if, for all $\delta > 0$, there exists $\tau_0 \in \mathbb{N}^+$ such that for all $\tau > \tau_0$:
\begin{align}
    \bigg{\lVert}\frac{1}{\sum_{t=t'}^{t'+\tau}\widetilde{N}^\varphi_i(t)}\sum_{t=t'}^{t'+\tau} \widetilde{N}^\varphi_i(t)\bx(t) - \frac{1}{\sum_{t=t'}^{t'+\tau}\widetilde{N}^\varphi_j(t)}\sum_{t=t'}^{t'+\tau} \widetilde{N}^\varphi_j(t)\bx(t) \bigg{\rVert}_1 < \delta  \qquad \, \forall \, t' \in \mathbb{N}^+, \, \forall \, i,j \in \mathcal{I}^{\varphi}. 
\end{align}
\end{definition}

{\cref{def:endo_fair_policy} implies that, for any endogenous group-fair policy, the following holds:
$$\lim_{T\to\infty}\frac{1}{\sum_{t=1}^{T}\widetilde{N}^\varphi_i(t)}\sum_{t=1}^{T} \widetilde{N}^\varphi_i(t)x_r(t) = \lim_{T\to\infty} \frac{1}{\sum_{t=1}^{T}\widetilde{N}^\varphi_j(t)}\sum_{t=1}^{T} \widetilde{N}^\varphi_j(t)x_r(t) \quad \forall i,j \in \mathcal{I}^{\varphi}, r \in \Xi.$$
}

{As before, we enforce the condition that, if two agents of different types are in the system in the same period, the distribution from which their rewards are drawn is identical (we henceforth refer to this as {\it individual fairness}). We refer to any policy that is both group-fair and individually fair as a fair policy. We then have the following theorem.
}

\begin{theorem}\label{thm:static-policies-are-opt-for-endo}
There exists an optimal fair policy that is static.
\end{theorem} 

\begin{proof}
Consider any fair policy $\varphi$ defined by the sequence $(\bx(t))_{t \in \mathbb{N}^+}$. We will construct a static policy whose long-run average profit lies within arbitrarily small $\epsilon > 0$ of that of $\varphi$. Taking $\epsilon$ to 0 will then prove the result. For ease of notation, in the remainder of the proof we omit the dependence of all quantities on $\varphi$.

Fix a type $i \in \mathcal{I}$, and construct the static policy $\widehat{\bx}$ as follows:
$$\widehat{x}_r = \lim_{T\to\infty} \sum_{t=1}^T \frac{\widetilde{N}_i(t)x_r(t)}{\sum_{t=1}^T \widetilde{N}_i(t)}, \quad \forall \ r \in \Xi.$$
For $j \in [K]$, let $\widehat{N}_j(t)$ denote the counterfactual number of type $j$ agents in the deterministic system in period $t$ under $\widehat{\bx}$, {and define $\widehat{N}_j = \lim_{T\to\infty}\frac1T\sum_{t=1}^T\widehat{N}_j(t)$. Since $\widehat{\bx}$ is static, it is easy to verify that
$\widehat{N}_j = \frac{\lambda_j}{\sum_r \ell_j(r)x_r}.$}

We first argue that the same set of types joins under $\widehat{\bx}$. To see this, for all type $j$ agents in $\mathcal{I}$, the counterfactual expected earnings under $\widehat{\bx}$ are:
\begin{align}\label{eq:same-set-joins}
\lim_{T\to \infty} \sum_{t=1}^T \frac{\widehat{N}_j(t)\sum_r rx_r(t)}{\sum_{t=1}^T \widehat{N}_j(t)} =  \sum_r r \widehat{x}_r &= \sum_r r \lim_{T\to\infty}  \sum_{t=1}^T \frac{\widetilde{N}_i(t)x_r(t)}{\sum_{t=1}^T \widetilde{N}_i(t)} = \lim_{T\to\infty}\sum_{t=1}^T \frac{\widetilde{N}_i(t) \sum_r rx_r(t)}{\sum_{t=1}^T \widetilde{N}_i(t)},
\end{align}
where the first equality follows from the fact that $\widehat{\bx}$ is static, the second from the definition of $\widehat{x}_r$, and the third uses the dominated convergence theorem for the sum-limit interchange. Noting that the final expression of \eqref{eq:same-set-joins} corresponds precisely to the expected earnings of type $j$ under $\varphi$ (the left-hand side of \eqref{eq:joining-condition}), we obtain that type $j$ joins under $\widehat{\bx}$ if and only if they joined under the original policy. Hence, it suffices to compare the profit of both policies over the set $\mathcal{I}$. Let $\widehat{N} = \sum_{j\in\mathcal{I}}\widehat{N}_j$, and $\widetilde{N}(t) = \sum_{j\in\mathcal{I}}\widetilde{N}_j(t)$. For ease of notation, we omit the indexing over $\mathcal{I}$ in the remainder of the proof.

The difference in the two policies' long-run average profit is given by:
{
\begin{align}\label{eq:delta-fix}
 &\left(\lim_{T\to\infty}\frac1T\sum_{t=1}^TR(\widetilde{N}(t))-\sum_r rx_r(t)\widetilde{N}(t)\right) -\left(\lim_{T\to\infty}\frac1T\sum_{t=1}^TR(\widehat{N}(t))-\sum_r r\widehat{x}_r\widehat{N}(t)\right) \notag \\
&\leq\left(\lim_{T\to\infty}\frac1T\sum_{t=1}^TR(\widetilde{N}(t))-\sum_r rx_r(t)\widetilde{N}(t)\right) -\left(R(\widehat{N})-\sum_r r\widehat{x}_r \widehat{N}\right) + (L+r_{\max})\lim_{T\to\infty}\frac1T\sum_{t=1}^T|\widehat{N}-\widehat{N}(t)|
\notag \\&= \left(\lim_{T\to\infty}\frac1T\sum_{t=1}^TR(\widetilde{N}(t))-\sum_r rx_r(t)\widetilde{N}(t)\right) -\left(R(\widehat{N})-\sum_r r\widehat{x}_r \widehat{N}\right).
\end{align}
where the first inequality follows from Lipschitz continuity of $R$, and the second inequality follows from the fact that $\widehat{N}(t)$ converges to $\lim_{T\to\infty}\frac1T\sum_{t=1}^T\widehat{N}(t) = \widehat{N}$.
}

Since $R$ is concave, by Jensen's inequality, we have:
\begin{align*}
\frac1T\sum_{t=1}^TR(\widetilde{N}(t)) \leq R\left(\frac1T\sum_{t=1}^T\widetilde{N}(t)\right) \leq R(\widehat{N}) + L|\widehat{N}-\frac1T\sum_{t=1}^T\widetilde{N}(t)| \leq R(\widehat{N}) + L\sum_j|\widehat{N}_j-\frac1T\sum_{t=1}^T\widetilde{N}_j(t)|,
\end{align*}
where the second inequality follows from the Lipschitz assumption. We use this to bound the difference in long-run average revenues:
\begin{align}\label{eq:endo-bound1}
\lim_{T\to\infty} \frac1T\sum_{t=1}^T R(\widetilde{N}(t)) - R(\widehat{N}) \leq L\sum_j |\widehat{N}_j-\lim_{T\to\infty}\frac1T\sum_{t=1}^T \widetilde{N}_j(t)|. 
\end{align}

The following lemma bounds the right-hand side of \eqref{eq:endo-bound1}. We defer its proof to the end of this section.

\begin{lemma}\label{lem:endo-bound1}
Under the $\mathbf{\widehat{x}}$ construction, for all $j \in \mathcal{I}$:
$$\lim_{T\to\infty} \frac1T \sum_{t=1}^T \widetilde{N}_j(t) = \widehat{N}_j.$$
\end{lemma}
Putting this together with \eqref{eq:endo-bound1} and plugging into \eqref{eq:delta-fix}, we have:
\begin{align}
\eqref{eq:delta-fix} &\leq  \sum_r r\widehat{x}_r\widehat{N} - \lim_{T\to\infty} \frac1T \sum_{t=1}^T \sum_r rx_r(t) \widetilde{N}(t) \notag \\
&= \sum_r r\left[{\widehat{N} \lim_{T\to\infty}\sum_{t=1}^T \frac{\widetilde{N}_i(t)x_r(t)}{\sum_{t=1}^T \widetilde{N}_i(t)} - \lim_{T\to\infty} \frac {\sum_{t=1}^T \widetilde{N}(t)}{T} \cdot \sum_{t=1}^T \frac{\widetilde{N}(t) x_r(t)}{\sum_{t=1}^T \widetilde{N}(t)}} \right] \label{eq:endo-bound-0}\\
&= \widehat{N}\sum_r r\left[\underbrace{\lim_{T\to\infty}\sum_{t=1}^T \frac{\widetilde{N}_i(t)x_r(t)}{\sum_{t=1}^T \widetilde{N}_i(t)} - \sum_{t=1}^T \frac{\widetilde{N}(t) x_r(t)}{\sum_{t=1}^T \widetilde{N}(t)}}_{(I)} \right] \label{eq:endo-bound-1}
\end{align}
where \eqref{eq:endo-bound-0} follows from the definition of $\widehat{x}_r$ and re-writing the third term, and \eqref{eq:endo-bound-1} follows from \cref{lem:endo-bound1}. We have the following lemma, whose proof we defer to the end of the section.
\begin{lemma}\label{lem:endo-bound-2}
For all $r \in \Xi$,
$$\sum_{t=1}^T \frac{\widetilde{N}(t)x_r(t)}{\sum_{t=1}^T \widetilde{N}(t)} \geq \min_{j\in \mathcal{I}} \sum_{t=1}^T \frac{\widetilde{N}_j(t)x_r(t)}{\sum_{t=1}^T \widetilde{N}_j(t)}.$$
\end{lemma}
Plugging this back into \eqref{eq:endo-bound-1}, we obtain:
\begin{align}\label{eq:last-step-endo}
(I) \leq \widehat{N}\cdot \lim_{T\to\infty}\left(\sum_{t=1}^T \frac{\widetilde{N}_i(t)x_r(t)}{\sum_{t=1}^T \widetilde{N}_i(t)} - \min_j\sum_{t=1}^T \frac{\widetilde{N}_j(t) x_r(t)}{\sum_{t=1}^T \widetilde{N}_j(t)}\right).
\end{align}
{
\cref{def:endo_fair_policy} implies that, for all $\delta > 0$, there exists large enough $\tau_0$ such that, for all $T > \tau_0$:
\begin{align}\label{eq:endo-fair-policy}
   \bigg{|}\frac{1}{\sum_{t=1}^{T}\widetilde{N}_i(t)}\sum_{t=1}^{T} \widetilde{N}_i(t))x_{r}(t) - \frac{1}{\sum_{t=1}^{T}\widetilde{N}_j(t)}\sum_{t=1}^{T} \widetilde{N}_j(t)x_r(t) \bigg{|} < \delta \quad \, \forall \, i,j \in [K], r \in \Xi. 
\end{align}
Applying this to \eqref{eq:last-step-endo} and taking $\delta\to 0$, we obtain that the static construction $\widehat{\mathbf{x}}$ weakly improves upon its dynamic counterpart, implying optimality of static policies.
}
\hfill\Halmos\end{proof}


\subsubsection{Auxiliary lemmas}
\begin{proof}[Proof of \cref{lem:endo-bound1}.]
By the recursive update equations, for all $t\geq 1$, $j \in \mathcal{I}$:
\begin{align*}
&\widetilde{N}_j(t+1) = \widetilde{N}_j(t) + \lambda_j  - \widetilde{N}_j(t)\sum_r\ell_j(r)x_r(t)\\
\implies &\frac{\widetilde{N}_j(T+1)-\widetilde{N}_j(1)}{\sum_{t=1}^T \widetilde{N}_j(t)} = \frac{T\lambda_j}{\sum_{t=1}^T \widetilde{N}_j(t)} - \sum_r\ell_j(r)\cdot \frac{\sum_{t=1}^T\widetilde{N}_j(t)x_r(t)}{\sum_{t=1}^T \widetilde{N}_j(t)}\\
\implies &\frac{\sum_{t=1}^T \widetilde{N}_j(t)}{T} = \lambda_j\left(\frac{1}{\frac{\widetilde{N}_j(T+1)-\widetilde{N}_j(1)}{\sum_{t=1}^T \widetilde{N}_j(t)} + \sum_r \ell_j(r)\cdot\frac{\sum_{t=1}^T\widetilde{N}_j(t)x_r(t)}{\sum_{t=1}^T \widetilde{N}_j(t)}}\right) \\
\implies &\lim_{T\to\infty}\frac{\sum_{t=1}^T \widetilde{N}_j(t)}{T} = \frac{\lambda_j}{\sum_r \ell_j(r)\cdot\lim_{T\to\infty}\frac{\sum_{t=1}^T\widetilde{N}_j(t)x_r(t)}{\sum_{t=1}^T \widetilde{N}_j(t)}} =\frac{\lambda_j}{\sum_r \ell_j(r)\cdot\lim_{T\to\infty}\frac{\sum_{t=1}^T\widetilde{N}_i(t)x_r(t)}{\sum_{t=1}^T \widetilde{N}_i(t)}},
\end{align*}
where the final equality follows from the group-fairness definition (\cref{def:endo_fair_policy}).
Recalling that, for any $j \in \mathcal{I}$:
\begin{align*}
\widehat{N}_j = \frac{\lambda_j}{\sum_r \ell_j(r)\widehat{x}_r} = \frac{\lambda_j}{\sum_r \ell_j(r)\cdot\lim_{T\to\infty}\frac{\sum_{t=1}^T\widetilde{N}_i(t)x_r(t)}{\sum_{t=1}^T \widetilde{N}_i(t)}},
\end{align*}
we obtain the lemma.
\hfill\Halmos
\end{proof}

\begin{proof}[Proof of \cref{lem:endo-bound-2}.]
Suppose for contradiction the inequality does not hold. Then, for all $j \in \mathcal{I}$:
\begin{align}
 \frac{\sum_{t=1}^T\widetilde{N}(t)x_r(t)}{\sum_{t=1}^T \widetilde{N}(t)} &<  \frac{\sum_{t=1}^T\widetilde{N}_j(t)x_r(t)}{\sum_{t=1}^T \widetilde{N}_j(t)} \\
\implies  \frac{\sum_{t=1}^T\widetilde{N}_j(t)}{\sum_{t=1}^T \widetilde{N}(t)} &<  \frac{\sum_{t=1}^T\widetilde{N}_j(t)x_r(t)}{\sum_{t=1}^T \widetilde{N}(t)x_r(t)} \\
\implies \sum_j  \frac{\sum_{t=1}^T\widetilde{N}_j(t)}{\sum_{t=1}^T \widetilde{N}(t)} &< \sum_j \frac{\sum_{t=1}^T\widetilde{N}_j(t)x_r(t)}{\sum_{t=1}^T \widetilde{N}(t)x_r(t)}.
\end{align}
Interchanging the sums on both sides, we obtain $1 < 1$, a contradiction.
\hfill\Halmos
\end{proof}
}

\subsection{Explicit discrimination: Proof of \cref{prop:explicit-wage-disc}}

For a given distribution $\bx$, let $\widehat{\ell}_i(\bx) = \sum_r \ell_i(r) x_r$, for $i \in [K]$, and $\widehat{r}(\bx) = \sum_r rx_r$.

\begin{proof}
The proof is constructive. Consider the {belief-based} policy $\varphi^b$ defined in \cref{alg:bbp}.

\begin{algorithm}
\begin{algorithmic}
\For{$t\in\mathbb{N}^+$}
\If{$\widetilde{N}(t-1) + \lambda_1 + \lambda_2 < D$}
\State Pay all new arrivals $v_1$.
\State Label any agent that stays in the system a type 1 agent; and pay all such agents $v_1$ 
\State for all $t' > t$.
\EndIf
\State Let $t^* = \inf\left\{t\, : \, \widetilde{N}(t-1) + \lambda_1 + \lambda_2 \geq D\right\}$.
\If{$t \geq t^*$}
\State Pay $D-\left(\lambda_1 + \lambda_2\right)$ type 1 agents identified a reward $v_1$, and everyone else a reward 0.
\EndIf
\EndFor
\end{algorithmic}
	\caption{Belief-based policy $\varphi^b$\label{alg:bbp}}
\end{algorithm}

We first derive the long-run average profit of $\varphi^b$:
\begin{align*}
    \widetilde{\Pi}(\varphi^b) &= \lim_{T\to\infty} \frac1T \sum_{t=1}^T \alpha\min\{\widetilde{N}(t-1)+\lambda_1 + \lambda_2,D\} - \left(\widetilde{N}(t-1)+\lambda_1 + \lambda_2\right)\sum_r rx_r(t) \\ 
    &= \lim_{T\to\infty}\frac1T \sum_{t=1}^{t^*-1} \left[\left(\alpha-\sum_r rx_r(t)\right) \left(\widetilde{N}(t-1)+\lambda_1 + \lambda_2\right) +\sum_{t=t^*}^T\left[\alpha D - v_1(D-(\lambda_1 + \lambda_2))\right]\right] \\
    &= \alpha D - v_1(D-(\lambda_1 + \lambda_2)),
\end{align*}
where the final equality follows from the fact that $t^* \leq 4$. The following lemma helps to characterize the optimal objective attained by any static policy. 
\begin{lemma}\label{lem:static-keeps-d-on}
Suppose $\alpha > 2v_2$, and let $\widetilde{N}^s$ denote the number of agents induced by the optimal static solution. Then $\widetilde{N}^s = D$.
\end{lemma}

 Thus, it suffices to show that
    $\min\left\{\widehat{r}(\bx^{s_1}),\widehat{r}(\bx^{s_2})\right\}\cdot D - v_1(D-(\lambda_1 + \lambda_2)) = \Omega(D)$,
where $\bx^{s_i}$ denotes the static policy with support $\{0,v_i\}$ that induces $D$ agents in the system, for $i \in \{1,2\}$. Let $x^{s_i}$ denote the weight placed on $v_i$ by $\bx^{s_i}$, for $i \in \{1,2\}$. (As noted in the proof of \cref{lem:static-keeps-d-on}, it is without loss of generality to consider only these two types of distributions.)

Consider first $\bx^{s_1}$. $x^{s_1}$ satisfies:
    $\frac{\lambda_1}{1-x^{s_1}} + \lambda_2 = D \iff x^{s_1} = 1 - \frac{\lambda_1}{D-\lambda_2}.$
After some algebra, we have $ \widehat{r}(\bx^{s_1})D-v_1(D-(\lambda_1 + \lambda_2)) = v_1D/4$.

Consider now $\bx^{s_2}$. Similarly, we have:
$    \frac{\lambda_1 + \lambda_2}{1-x^{s_2}} = D \iff x^{s_2} = 1-\frac{\lambda_1 + \lambda_2}{D}.$
Then, algebraic manipulations give us $ \widehat{r}(\bx^{s_2})D-v_1(D-(\lambda_1 + \lambda_2)) = (v_2-v_1)D/4.$
\hfill\Halmos\end{proof}

\subsubsection{Auxiliary lemmas}\label{apx:aux-lemmas-belief-based}

\begin{proof}[Proof of \cref{lem:static-keeps-d-on}.]
For $i\in \{1,2\}$, let $\varphi^{(i)}$ denote the reward distribution with support $\{0,v_i\}$ that places weight $x_i$ on $v_i$, and $\widetilde{\Pi}(\varphi^{(i)})$ its associated profit. Note that only considering distributions with these two supports is without loss of generality, as $(i)$ we have established that there exists an optimal solution with a support of size 2, and $(ii)$ $\bx$ such that $\supp(\bx) = \{v_1,v_2\}$ induces $\widetilde{N}_1 = \infty$ agents, with only $\alpha D$ revenue.

It suffices to show that, for $i \in \{1,2\}$, $\frac{d}{dx_i}\widetilde{\Pi}(\varphi^{(i)}) > 0$ for all $x_i$ such that $\widetilde{N}^{\varphi^{(i)}} < D$, where $\widetilde{N}^{\varphi^{(i)}}$ is used to denote the number of agents induced by $\varphi^{(i)}$.

For $i \in \{1,2\}$, we have:
    $\frac{d}{dx_i}\widetilde{\Pi}(\varphi^{(i)}) = \frac{d}{dx_i} \left[(\alpha - v_ix_i) \widetilde{N}^{\varphi^{(i)}}\right] 
    = -v_i\widetilde{N}^{\varphi^{(i)}} + (\alpha-v_ix_i)\frac{d}{dx_i}\widetilde{N}^{\varphi^{(i)}}.$
Thus,
\begin{align}\label{eq:static-keeps-d-on}
    \frac{d}{dx_i}\widetilde{\Pi}(\varphi^{(i)}) > 0 \iff \alpha > v_i\left (x_i + \frac{\widetilde{N}^{\varphi^{(i)}}}{\frac{d}{dx_i}\widetilde{N}^{\varphi^{(i)}}}\right)
\end{align}

For $i = 1$, we have:
\begin{align*}
    \widetilde{N}^{\varphi^{(1)}} = \frac{\lambda}{1-x_1} + \lambda \implies & \frac{d}{dx_1}\widetilde{N}^{\varphi^{(1)}} = \frac{\lambda}{(1-x_1)^2} 
    \implies  \frac{\widetilde{N}^{\varphi^{(1)}}}{\frac{d}{dx_1}\widetilde{N}^{\varphi^{(1)}}} = (1-x_1) + (1-x_1)^2
\end{align*}

For $i = 2$:
\begin{align*}
    \widetilde{N}^{\varphi^{(2)}} = \frac{2\lambda}{1-x_2} \implies & \frac{d}{dx_2}\widetilde{N}^{\varphi^{(2)}} = \frac{2\lambda}{(1-x_2)^2} 
    \implies  \frac{\widetilde{N}^{\varphi^{(2)}}}{\frac{d}{dx_2}\widetilde{N}^{\varphi^{(2)}}} = 1-x_2
\end{align*}

Noting that $\alpha > 2v_2$ satisfies \eqref{eq:static-keeps-d-on} for all $x_i$, $i \in \{1,2\}$, we obtain the lemma.
\hfill\Halmos\end{proof}

\subsection{Implicit discrimination: Proofs}\label{apx:cyclic-diff-dist}

In this section, we use the wrap-around convention that, given a cyclic policy, for $t' \in \{-(\tau-1),\ldots,0\}$, $\bx(t') := \bx(\tau+t')$.

We  prove \cref{ex:steady-state-cyclic}.
\begin{proof}[Proof of \cref{ex:steady-state-cyclic}.]

\chedit{We obtain $\widetilde{N}_i(t)$ by solving the following system of equations:
 \begin{align*}
\begin{cases}
\widetilde{N}_i(1) &= \widetilde{N}_i(2) + \lambda_i - \widetilde{N}_i(2)\sum_{r'} \ell_i(r')x_{r'}(2) \\
\widetilde{N}_i(2) &= \widetilde{N}_i(1) + \lambda_i - \widetilde{N}_i(1)\sum_{r'} \ell_i(r')x_{r'}(1)
\end{cases}
 \end{align*}
Hence, we have that for all $t$ odd:
\begin{align}\label{eq:steadystate-cyclic1}
    &\widetilde{N}_1(t) = 1.9\lambda,
    \widetilde{N}_2(t) = \lambda 
    \implies  \widetilde{N}(t) = 2.9\lambda.
\end{align}

Similarly, for all $t$ even we have:
\begin{align}\label{eq:steadystate-cyclic2}
    \widetilde{N}_1(t) = 2\lambda,
    \widetilde{N}_2(t) = 1.5\lambda
    \implies \widetilde{N}(t) = 3.5\lambda.
\end{align}
}

We leverage \eqref{eq:steadystate-cyclic1} and \eqref{eq:steadystate-cyclic2} for the following two results. In particular, \cref{lem:exists-rev-counter} characterizes the optimal static solution. \cref{lem:cyclic-beats-zero-counter} establishes that the cyclic policy outperforms the static solution.
\begin{lemma}\label{lem:exists-rev-counter}
For $\alpha < r$, the optimal static solution $\bx^s$ fulfills $x^s = 0$, where $x^s$ is the weight on reward $r > 0$.
\end{lemma}

\begin{lemma}\label{lem:cyclic-beats-zero-counter}
For $\alpha \geq 0.7r$, $\widetilde{\Pi}(\varphi^c) - \widetilde{\Pi}(\varphi^0)  = \Omega(\lambda)$, where $\varphi^0$ denotes the policy that deterministically pays out zero reward in each period.
\end{lemma}

 Putting these two lemmas together proves the result.
 \hfill\Halmos\end{proof}

 \chedit{The following proposition shows the difference in conditional reward distribution across both types.

 \begin{proposition}\label{prop:example-disc}
Under $\varphi^c$, a type 1 agent receives reward $r$ with probability 19/39, and a type 2 agent receives reward $r$ with probability 0.4. Hence, the expected reward of a type 1 agent is $\frac{19}{39}r$, and the expected reward of a type 2 agent is $0.4r$.
 \end{proposition}

 \begin{proof}
By \eqref{eq:steadystate-cyclic1} and \eqref{eq:steadystate-cyclic2}, the probability that a type 1 agent receives reward $r$, conditional on being a type 1 agent is given by:
\begin{align*}
\frac{\widetilde{N}_1(1)}{\widetilde{N}_1(1) + \widetilde{N}_1(2)} = 19/39.
\end{align*}
On the other hand, the probability that a type 2 agent receives reward $r$, conditional on being a type 2 agent is given by:
\begin{align*}
\frac{\widetilde{N}_2(1)}{\widetilde{N}_2(1) + \widetilde{N}_2(2)} = 2/5.
\end{align*}
The expected reward for each type follows from the fact that in every odd period, a reward of 0 is paid out.
 \end{proof}
}

\subsubsection{Auxiliary lemmas}\label{apx:cyclic-aux-proofs}

\begin{proof}[Proof of \cref{lem:exists-rev-counter}.]
We abuse notation and let $\widetilde{N}(\bx)$ and $\widetilde{\Pi}(\bx)$ respectively denote the number of agents and profit induced by static solution $\bx$ that places weight $x$ on $r$. It suffices to show that $\frac{d}{dx}\widetilde{\Pi}(\bx) < 0$ for all $x \in (0,1)$. We have:
\begin{align}
    \frac{d}{dx}\widetilde{\Pi}(\bx) &= -r\widetilde{N}(\bx) + (\alpha-rx)\frac{d}{dx}\widetilde{N}(\bx)
    < -r\widetilde{N}(\bx) + r(1-x)\frac{d}{dx}\widetilde{N}(\bx) \label{eq:lemc4-step1}
    \end{align}
where \eqref{eq:lemc4-step1} follows from the assumption that $\alpha < r$ and $\frac{d}{dx}\widetilde{N}(\bx) > 0$. Consider now our specific instance. We have:
\begin{align*}
    \widetilde{N}(\bx) &= \frac{0.1\lambda}{0 \cdot x + 0.1 (1-x)} + \frac{\lambda}{0.5x + 1-x} 
    \implies \frac{d}{dx}\widetilde{N}(\bx) = \lambda\left(\frac{1}{(1-x)^2} + \frac{0.5}{(1-0.5x)^2}\right).
\end{align*}

Plugging this into \eqref{eq:lemc4-step1}, we obtain:
\begin{align}
    \frac{d}{dx}\widetilde{\Pi}(\bx) &< \lambda r \left[-\left(\frac{1}{1-x} + \frac{1}{1-0.5x}\right) + (1-x)\left(\frac{1}{(1-x)^2} + \frac{0.5}{(1-0.5x)^2}\right)\right] \notag \\
    &= \lambda r\left[\frac{-\left(1-0.5x\right) + 0.5(1-x)}{(1-0.5x)^2}\right] 
    < 0 \quad \forall \, x \in (0,1). \notag
\end{align}
\hfill\Halmos\end{proof}

\begin{proof}[Proof of \cref{lem:cyclic-beats-zero-counter}.]
We have:
\begin{align*}
    \widetilde{\Pi}(\varphi^c) - \widetilde{\Pi}(\varphi^0) &= \left(\alpha\left(\frac12\widetilde{N}(1) + \frac12\widetilde{N}(2)\right) - \frac12 r \widetilde{N}(1) \right) - \alpha(\lambda_1+\lambda_2)  \geq 0.7r\left(3.2\lambda-1.1\lambda\right) - \frac12r \cdot 2.9\lambda 
    = 0.02r\lambda,
\end{align*}
where the inequality follows from plugging in the expressions for $\widetilde{N}(t), t \in \{1,2\}$, as well as $\alpha \geq 0.7r$.
\hfill\Halmos\end{proof}

\section{Structure and computation of the fluid heuristic}\label{apx:structure-of-fluid}

\subsection{Proof of Proposition~\ref{prop:fluid-non-cvx}}

\begin{proof}
{Suppose} $\rev(\fluidn) = \fluidn$, $\rewardset = \{r, 0\}$ for some $r > 0$, and $K = 1$. Moreover, let $\ell_1(0) = 1$. For distribution $\bx$, let $x$ denote the weight placed on $r$. We abuse notation and let $\widehat{\Pi}(\bx)$ denote the profit induced by $\bx$ in the deterministic relaxation. The objective of~\ref{eq:fluid-opt} evaluates to:
$$\widehat{\Pi}(\bx) = \lambda_1 \, \frac{1-rx}{1 - (1-\ell_1(r))x} \implies \frac{\partial^2}{\partial x^2}\widehat{\Pi}(\bx) = 2\lambda_1 \, \frac{(1-\ell_1(r))(1-\ell_1(r)-r)}{(1 - (1-\ell_1(r))x)^3} > 0 \quad  \forall \, r < 1-\ell_1(r).$$ \hfill\Halmos
\end{proof}

\subsection{Proof of \cref{thm:main-theorem}}

{Consider the following {\it budgeted supply maximization} problem, which we term~\ref{eq:supply-opt}. {For a fixed set of types, an instance of~\ref{eq:supply-opt} is defined by a budget $B \in \mathbb{R}_{+}$ as follows:}}
\begin{align}\label{eq:supply-opt}
    \max_{{\xvec \in \simplex^{|\rewardset|}, \mathbf{\fluidn} \in \mathbb{N}^K}} \quad & \sum_i \fluidn_i \tag{\supplyopt} \\
    \text{subject to} \quad & \left(\sum_r rx_r\right)\left(\sum_i \fluidn_i\right) \leq B \notag\\
    & \fluidn_i=\frac{\lambda_i}{\sum_r \ell_i(r) x_r}\notag
\end{align}

{Before proceeding with the proof, we note that the introduction of~\ref{eq:supply-opt} is a {\it conceptual} simplification, rather than a computational one. In particular, though the objective is now linear in $\mathbf{\fluidn}$, the budget constraint remains non-convex.}
{Lemma~\ref{prop:profit-reduction} formalizes the idea that~\ref{eq:supply-opt} is a useful proxy via which one can characterize the fluid optimum of~\ref{eq:fluid-opt}. Namely, it establishes that that the optimal solution to \ref{eq:fluid-opt} inherits the same structure as that of~\ref{eq:supply-opt}. We defer its proof -- as well as the proofs of all subsequent lemmas -- to Appendix \ref{apx:main-thm-aux-lemmas}.}
\begin{lemma}\label{prop:profit-reduction}
Suppose there exists $n \in \mathbb{N}^+$ such that, for all instances of~\ref{eq:supply-opt}, there exists an optimal solution $\supplyoptsol$ such that $|\supp(\nf{\supplyoptsol})| \leq n$. Then, there exists an optimal solution $\nf{\bx}^\star$ to~\ref{eq:fluid-opt} such that $|\supp(\nf{\bx}^\star)| \leq n$. 
\end{lemma}

{Thus, to prove the theorem, it suffices to show {that for all instances of~\ref{eq:supply-opt}} there exists {an optimal solution $\supplyoptsol$ with the desired structure, i.e. $|\supp(\nf{\supplyoptsol)}| \leq 2$.}}
Let $\supplyoptsol$ denote an optimal solution to~\ref{eq:supply-opt}. {Based on the second constraint}, we define $\fluidn_i(\bx) = \frac{\lambda_i}{\sum_r \ell_i(r)x_r}$, and $\fluidn(\bx) = \sum_i \fluidn_i(\bx)$. {Further}, $\avgreward(\bx) = \sum_r rx_r$ is used to denote {an agent's} expected reward under $\bx$, and $\avgloss_i(\bx) = \sum_r \ell_i(r)x_r$ the expected departure probability of a type $i$ agent under~$\bx$. {We moreover assume that $B$ is such that $\rmin\fluidn(\unitvec_{\rmin}) < B$ (i.e., the budget is large enough to serve all agents with the minimum possible reward).}
Finally, we introduce the concept of \emph{interlacing rewards}. {Intuitively, two rewards $r<r'$ {interlace} the budget $B$ if a solution $\xvec=\unitvec_r$ spends at most~$B$, while $\xvec'=\unitvec_{r'}$ spends more than $B$.}

\begin{definition}[Interlacing rewards]
Consider an instance of~\ref{eq:supply-opt}. Rewards $r$ and $r'$, $r < r'$, are said to \emph{interlace} budget $B$ if the following holds:
$ {r \fluidn(\unitvec_{r}) \leq B < r'\fluidn(\unitvec_{r'}). }$
We also say that three rewards $r, r'$ and $r''$ interlace if there exist two pairs of interlacing rewards amongst $r, r'$ and $r''$.
\end{definition}

{We present a simple, but useful, lemma upon which we repeatedly rely to prove Theorem~\ref{thm:main-theorem}. We leave the easy proof of this fact to the reader.}

\begin{lemma}\label{lem:obvious-fact}
Suppose distribution $\by$ is such that $\supp(\by) = \left\{r, r'\right\}$ for some $r > r' \in \rewardset$, with weights $y$ on $r$ and $1-y$ on $r'$, $y \in (0,1)$. Then, $\avgreward(\by)$ and $\fluidn(\by)$ are {both} non-decreasing in $y$. 
\end{lemma}

With these two lemmas in hand, we prove the theorem.

\begin{proof}[Proof of \cref{thm:main-theorem}.]
We first consider the trivial case in which $\rmax\fluidn(\unitvec_{\rmax}) \leq B$. In this case, {the optimal solution for \ref{eq:supply-opt} is} to give out the maximum reward with probability~1 (i.e., $\bx = \unitvec_{\rmax}$). 
Thus, in the remainder of the proof, we assume $\rmax\fluidn(\unitvec_{\rmax}) > B$.

Any optimal solution $\supplyoptsol$ to \ref{eq:supply-opt} satisfies the budget constraint with equality (else, we can strictly improve the solution by moving weight from a low reward to a high reward, contradicting optimality of $\supplyoptsol$). As a result, we can instead consider the following equivalent {\it budgeted reward minimization} problem:
\begin{align}
    \min_{{\xvec \in \simplex^{|\rewardset|}}} \qquad &\sum_r rx_r \tag{\textsc{Reward-OPT}} \label{eq:reward-opt} \\
    \text{subject to} \qquad &\left(\sum_r r x_r\right)\left(\sum_i \fluidn_i\right) {=} B \label{eq:budget-constraint}\\
    & \fluidn_i=\frac{\lambda_i}{\sum_r \ell_i(r) x_r}\notag 
\end{align}

We abuse notation and let $\supplyoptsol$ denote an optimal solution to~\ref{eq:reward-opt}, and suppose $|\supp(\supplyoptsol)| > 2$. Algorithm~\ref{alg:support-reduction} exhibits a procedure which either contradicts optimality of $\supplyoptsol$, or iteratively reduces the size of the support of $\supplyoptsol$, all the while maintaining feasibility  -- i.e., satisfying~\eqref{eq:budget-constraint} -- and never degrading the quality of the solution. Thus, {we only need to show} 
that Algorithm~\ref{alg:support-reduction} maintains these two invariants and successfully terminates. {As stated, the first step of each iteration of the algorithm implicitly relies on the existence of three interlacing rewards. It is, however, \emph{a priori} unclear that such a set necessarily exists. Lemma~\ref{lem:opt-uses-interlacing-rewards} establishes that, as long as the incumbent solution $\algxvec$ is feasible, existence of three interlacing rewards is{,} in fact{,} guaranteed.}


\begin{algorithm}
\begin{algorithmic}
\Require {optimal solution $\supplyoptsol$ to~\ref{eq:reward-opt} such that $|\supp(\supplyoptsol)| > 2$}
\Ensure {optimal solution $\algxvec$ such that $|\supp(\algxvec)| = 2$, or that $\supplyoptsol$ is not optimal, a contradiction.}
\State $\algxvec \gets \supplyoptsol$
\While{$|\supp(\algxvec)| > 2$}
\State Choose a set $\mathcal{R} = \{r_1, r_2, r_3\} \subseteq \supp(\algxvec)$ of interlacing rewards.
\State 
Construct a distribution $\by$ such that $\supp(\by) = \supp(\algxvec) \setminus \{r^\star\}, r^\star\in \mathcal{R}$, such that
\State $(i)$ $\avgreward(\by) = \avgreward(\algxvec)$ and $(ii)$ $\fluidn(\by) \geq \fluidn(\algxvec)$.
\If{$\fluidn(\by) > \fluidn(\algxvec)$}
\State Return that $\supplyoptsol$ is not optimal for~\ref{eq:reward-opt}, a contradiction.
\EndIf
\State $\algxvec \gets \by$
\EndWhile
\end{algorithmic}
	\caption{Support-reduction procedure}\label{alg:support-reduction}
\end{algorithm}


\begin{lemma}\label{lem:opt-uses-interlacing-rewards}
Suppose $\algxvec$ is feasible and $|\supp(\algxvec)| > 2$. Then, there exists a set $\mathcal{R} = \{r_1,r_2,r_3\}\subseteq\supp(\algxvec)$ of interlacing rewards.
\end{lemma}

{{Next, Lemma~\ref{lem:main-lemma}} establishes that, given an incumbent feasible solution $\algxvec$, it is {\it always} possible to (weakly) increase the fluid number of agents by removing a reward from the support of $\algxvec$, all the while holding the expected reward fixed. Moreover, one can do this such that at least two interlacing rewards remain in the support of the newly constructed distribution.}

\begin{lemma}\label{lem:main-lemma}
Consider a set $\mathcal{R} = \{r_1, r_2, r_3\} \subseteq \supp(\algxvec)$ with $r_1 > r_2 > r_3$ interlacing. Then, there exists a distribution $\by$ and $r^\star \in \mathcal{R}$ such that $(i)$ $\supp(\by) = \supp(\algxvec) \setminus \{r^\star\}$, $(ii)$ $\avgreward(\by) = \avgreward(\algxvec)$, $(iii)$ $\fluidn(\by) \geq \fluidn(\algxvec)$, and $(iv)$ $p, q \in \mathcal{R}\setminus\{r^\star\}$ are interlacing.
\end{lemma}

{Finally,} {Lemma~\ref{lem:support-reduction-maintains-invariants}} shows that either this new supply-improving solution $\by$ is optimal, or it implies the existence of a feasible solution which {\it strictly} improves upon $\supplyoptsol$, contradicting the assumption that $\supplyoptsol$ was optimal in the first place.
\begin{lemma}\label{lem:support-reduction-maintains-invariants}
At the end of each iteration, either $\algxvec$ is optimal, or the algorithm has correctly returned that $\supplyoptsol$ was not optimal to begin with.
\end{lemma}

Lemmas~\ref{lem:opt-uses-interlacing-rewards},~\ref{lem:main-lemma}, and~\ref{lem:support-reduction-maintains-invariants} together establish that, in each iteration of the procedure, both invariants are maintained. {It remains to show that Algorithm~\ref{alg:support-reduction} terminates. Suppose first that the algorithm returns that~$\supplyoptsol$ was not optimal to begin with. In this case, the algorithm clearly terminates. In the other case, termination follows since in each iteration of the algorithm, $|\supp(\algxvec)|$ strictly decreases.}\hfill\Halmos
\end{proof}

\subsection{Auxiliary lemmas}\label{apx:main-thm-aux-lemmas}

\subsubsection{Proof of Lemma~\ref{prop:profit-reduction}}\label{apx:proof-of-profit-red}

\begin{proof}
Consider an optimal solution $\bx^\star$ to~\ref{eq:fluid-opt}, and suppose $|\supp(\bx^\star)| > n$. We define $C\left(\fluidn(\bx^\star)\right)$ to be the minimum cost of securing a supply of $\fluidn(\bx^\star)$ agents in the system. Formally, 
$C\left(\fluidn(\bx^\star)\right) = \min_{\bx: \fluidn(\bx) = \fluidn(\bx^\star)} \fluidn(\bx)\avgreward(\bx).$
We moreover let $\widehat{\Pi}(\bx) = \rev(\fluidn(\bx)) - C\left(\fluidn(\bx)\right).$

Consider now an instance of~\ref{eq:supply-opt} with budget $B = C\left(\fluidn(\bx^\star)\right)$. Let $\supplyoptsol$ denote an optimal solution to this instance of~\ref{eq:supply-opt}, with $|\supp(\supplyoptsol)| \leq n$, and let $\fluidn(\supplyoptsol)$ denote the number of agents induced by $\supplyoptsol$. 
Clearly, $\bx^\star$ is feasible to this instance, and thus $\fluidn(\supplyoptsol)\geq\fluidn(\bx^\star)$. Since $\rev$ is increasing, we have $\rev(\fluidn(\supplyoptsol)) \geq \rev(\fluidn(\bx^\star))$. Thus,
\begin{align*}
    \widehat{\Pi}(\supplyoptsol) = \rev(\fluidn(\supplyoptsol)) - C\left(\fluidn(\supplyoptsol)\right) \geq \rev(\fluidn(\bx^\star))-C\left(\fluidn(\supplyoptsol)\right) \geq \rev(\fluidn(\bx^\star))-C\left(\fluidn(\bx^\star)\right) = \widehat{\Pi}(\bx^\star)
\end{align*}
where the final inequality follows from the fact that $C\left(\fluidn(\supplyoptsol)\right) \leq B = C\left(\fluidn(\bx^\star)\right)$, by construction. Thus, $\supplyoptsol$ is also optimal for~\ref{eq:fluid-opt}. Since $|\supp(\supplyoptsol)| \leq n$ by assumption, the result follows.\hfill\Halmos
\end{proof}

\subsubsection{Proof of \cref{lem:opt-uses-interlacing-rewards}}

\begin{proof}
By feasiblity of $\algxvec$ and the assumption that $|\supp(\algxvec)| > 2$, there exists $r \in \supp(\algxvec)$ such that $r\fluidn(\unitvec_{r}) < B$ (else the budget has been exceeded).  By the same argument, there exists $r \in \supp(\algxvec)$ such that $r\fluidn(\unitvec_r) > B$ (else the total reward paid out is {\it strictly} less than the budget, contradicting feasibility to~\ref{eq:reward-opt}). Let $r_1 = \inf\left\{r \in \supp(\algxvec) \big{|} r\fluidn(\unitvec_r) > B \right\}$, and $r_2 = \sup\left\{r \in \supp(\algxvec) \big{|} r\fluidn(\unitvec_r) < B \right\}$. Then, by definition, $(r_1, r_2)$ is an interlacing pair.

Moreover, since $|\supp(\algxvec)| > 2$, there necessarily exists $r_3$ such that, either $r_3\fluidn(\unitvec_{r_3}) \leq B$, or $r_3\fluidn(\unitvec_{r_3})> B$, as argued above. In the former case $(r_1,r_3)$ form an interlacing pair; in the latter case, $(r_3, r_2)$ form an interlacing pair. Having established that $(r_1,r_2)$ is an interlacing pair, in either case we obtain at least 2 interlacing pairs.\hfill\Halmos
\end{proof}

\subsubsection{Proof of Lemma~\ref{lem:main-lemma}}

\begin{proof}
The proof is constructive. Fix ${r}^\star \in \mathcal{R}$, and let $\by(p,q)$ be the distribution that distributes the weight placed on $r^\star$ by $\algxvec$ to $p, q \in \mathcal{R} \setminus \{r^\star\}$, $p\neq q$, all the while holding the expected reward constant. Formally, $\by(p,q)$ satisfies:
    (1) $y_{r} = \widetilde{x}_{r}^S$ for $r \not\in \mathcal{R}$.
    (2) $y_{r^\star} = 0$,
    (3) $y_{q} = \sum_{r\in\mathcal{R}}\algx_r - y_{p}$ for $p, q \in \mathcal{R}\setminus \{r^\star\}$, and
    (4) $\avgreward(\by(p,q)) = \avgreward(\algxvec)$.

Recall, $\mathcal{R} = \{r_1, r_2, r_3\} \subseteq \supp(\algxvec)$, with $r_1 > r_2 > r_3$ interlacing. We claim that one of the following two inequalities necessarily holds:
\begin{align}\label{ineq:main-lemma-2}
\fluidn\left(\by(r_1,r_3)\right) \geq \fluidn(\algxvec) \quad \text{ or } \quad \fluidn\left(\by(r_2,r_3)\right) \geq \fluidn(\algxvec),
\end{align}
{\it and} one of the following two inequalities necessarily holds:
\begin{align}\label{ineq:main-lemma-1}
\fluidn\left(\by(r_1,r_2)\right) \geq \fluidn(\algxvec) \quad \text{ or } \quad \fluidn\left(\by(r_1,r_3)\right) \geq \fluidn(\algxvec).
\end{align}
Suppose the claim is true. Then, the proof of the lemma is complete. To see that we have shown existence of $\by$ such that $p, q \in \mathcal{R}\setminus\{r^\star\}$ in the support of $\by$ are interlacing (i.e., Condition $(iv)$ of the lemma holds), we consider the following cases:
\begin{enumerate}
    \item $(r_1, r_3)$ and $(r_2, r_3)$ are the interlacing pairs:~\eqref{ineq:main-lemma-2} covers this case, as $\{r_1, r_3\}\subseteq \supp(\by(r_1,r_3))$, and $\{r_2, r_3\}\subseteq \supp(\by(r_2,r_3))$.
    \item $(r_1, r_3)$ and $(r_1, r_2)$ are the interlacing pairs:~\eqref{ineq:main-lemma-1} covers this case, as $\{r_1, r_2\} \subseteq \supp(\by(r_1,r_2))$, and $\{r_1, r_3\} \subseteq \supp(\by(r_1,r_3))$.
\end{enumerate}

We only show the proof of the first set of alternatives~\eqref{ineq:main-lemma-2}. The proof of~\eqref{ineq:main-lemma-1} is entirely analogous. With slight abuse of notation, we let $y_j = y_{r_j}, \algx_j = \algx_{r_j}$ for $r_j \in |\rewardset|$.  

Consider first distribution $\by(r_1,r_3)$. Solving $(iii)$ and $(iv)$ for $y_1$ and $y_3$, we obtain:
\begin{align*}
y_{1} = \algx_1 + \frac{r_2-r_3}{r_1-r_3}\algx_2 \quad , \quad y_3 = \algx_3 + \frac{r_1-r_2}{r_1-r_3} \algx_2.
\end{align*}

Since $\fluidn_i(\bx)$ is the composition of convex function $f(z) = \frac1z$ with affine mapping $g(x) = \sum_{r\in\rewardset} \ell_i(r)x_r$, we have that $\fluidn_i(\bx)$ is convex. Thus:
\begin{align*}
\fluidn_i(\by(r_1,r_3)) &\geq \fluidn_i(\algxvec) + \nabla \fluidn_i(\algxvec)^T(\by(r_1,r_3) - \algxvec) \\
&= \fluidn_i(\algxvec) - \frac{\lambda_i}{\hat{\ell}_i(\algxvec)^2}\begin{pmatrix}
\ell_{i}(r_1)\\
\ell_{i}(r_2) \\
\ell_i(r_3)
\end{pmatrix}^T	\begin{pmatrix}
\algx_1 + \frac{r_2-r_3}{r_1-r_3}\algx_2 - \algx_1\\
0-\algx_2 \\
\algx_3 + \frac{r_1-r_2}{r_1-r_3} \algx_2 - \algx_3
\end{pmatrix} \\
&= \fluidn_i(\algxvec) - \frac{\lambda_i}{\hat{\ell}_i(\algxvec)^2}\algx_2\left[\frac{r_2-r_3}{r_1-r_3}\ell_{i}(r_1)- \ell_{i}(r_2) + \frac{r_1-r_2}{r_1-r_3}\ell_{i}(r_3)\right].
\end{align*}

Summing over all $i$, we obtain that the following condition suffices for $\fluidn(\by(r_1,r_3)) \geq \fluidn(\algxvec)$ to hold:
\begin{align}\label{inequality-3}
\sum_i \frac{\lambda_i}{\hat{\ell}_i(\algxvec)^2}\left[\frac{r_2-r_3}{r_1-r_3}\ell_{i}(r_1)- \ell_{i}(r_2) + \frac{r_1-r_2}{r_1-r_3}\ell_{i}(r_3)\right] \leq 0.
\end{align}

Suppose~\eqref{inequality-3} fails to hold, i.e. $\sum_i \frac{\lambda_i}{\hat{\ell}_i(\algxvec)^2}\left[\frac{r_2-r_3}{r_1-r_3}\ell_{i}(r_1)- \ell_{i}(r_2) + \frac{r_1-r_2}{r_1-r_3}\ell_{i}(r_3)\right] > 0$. Consider now distribution $\by(r_2,r_3)$. Solving $(iii)$ and $(iv)$ for $y_{2}$ and $y_3$ we obtain:
$$y_2 = \algx_2 + \frac{r_1-r_3}{r_2-r_3}\algx_1 \quad , \quad y_3 = \algx_3 + \frac{r_2-r_1}{r_2-r_3}\algx_1.$$

Again, by convexity of $\fluidn(\bx)$:
\begin{align*}
\fluidn_i(\by(r_2,r_3)) &\geq \fluidn_i(\algxvec) + \nabla \fluidn_i(\algxvec)^T(\by(r_2,r_3) - \algxvec) \\
&= \fluidn_i(\algxvec) - \frac{\lambda_i}{\hat{\ell}_i(\algxvec)^2}\begin{pmatrix}
\ell_{i}(r_1)\\
\ell_{i}(r_2) \\
\ell_i(r_3)
\end{pmatrix}^T	\begin{pmatrix}
0-\algx_1\\
\algx_2 + \frac{r_1-r_3}{r_2-r_3}\algx_1 - \algx_2 \\
x_3^2 + \frac{r_2-r_1}{r_2-r_3}\algx_1 - x_3^2
\end{pmatrix} \\
&= \fluidn_i(\algxvec) - \frac{\lambda_i}{\hat{\ell}_i(\algxvec)^2}\algx_1\left[-\ell_{i}(r_1) + \frac{r_1-r_3}{r_2-r_3}\ell_{i}(r_2) + \frac{r_2-r_1}{r_2-r_3}\ell_i(r_3)\right].
\end{align*}

Summing over all $i$, the following condition suffices for $\fluidn(\by(r_2,r_3)) \geq \fluidn(\algxvec)$ to hold:
\begin{align*}
\frac{\lambda_i}{\hat{\ell}_i(\algxvec)^2}\algx_1\left[-\ell_{i}(r_1) + \frac{r_1-r_3}{r_2-r_3}\ell_{i}(r_2) + \frac{r_2-r_1}{r_2-r_3}\ell_i(r_3)\right] \leq 0 \\
\iff \frac{\lambda_i}{\hat{\ell}_i(\algxvec)^2}\left[\frac{r_2-r_3}{r_1-r_3}\ell_{i}(r_1) -\ell_{i}(r_2) + \frac{r_1-r_2}{r_1-r_3}\ell_i(r_3)\right] \geq 0
\end{align*}
which holds by assumption.\hfill\Halmos
\end{proof}

\subsubsection{Proof of Lemma~\ref{lem:support-reduction-maintains-invariants}}

\begin{proof}
By Lemma~\ref{lem:main-lemma}, there exists $\by$ with strictly smaller support than $\algxvec$ such that $\avgreward(\by) = \avgreward(\algxvec)$ and $\fluidn(\by) \geq \fluidn(\algxvec)$. 

Suppose first that $\fluidn(\by) = \fluidn(\algxvec)$. Then feasibility is clearly maintained, and $\by$ is also optimal. Suppose instead that $\fluidn(\by) > \fluidn(\algxvec)$. In this case, $\by$ is {\it not} feasible, as 
$\avgreward(\by)\fluidn(\by) > \avgreward(\algxvec)\fluidn(\algxvec) = B.$
We claim {we} 
can construct a {feasible} distribution $\bz$ such that $(i)$ $\supp(\bz) = \supp(\by)$, $(ii)$ $\avgreward(\bz) < \avgreward(\by)$, and $(iii)$ $\avgreward(\bz)\fluidn(\bz) = B$. Since $\avgreward(\by) = \avgreward(\algxvec)$ by construction, the existence of such a distribution $\bz$ would contradict optimality of $\algxvec$, and consequently that of $\supplyoptsol$. 

To complete the proof, we detail the construction of such a distribution $\bz$. Take interlacing rewards $r_{\max}(\by) = \sup \{r | r \in \supp(\by)\}$, $r_{\min}(\by) = \inf\{r | r \in \supp(\by)\}$ (of which Lemma~\ref{lem:main-lemma} guarantees the existence), and continuously move weight from $r_{\max}(\by)$ to $r_{\min}(\by)$, holding all else fixed. This procedure strictly decreases $\avgreward(\by)$ and weakly decreases $\fluidn(\by)$, by Lemma~\ref{lem:obvious-fact}, and thus strictly decreases $\avgreward(\by)\fluidn(\by)$. Do this until one of two events occurs:
\begin{enumerate}
\item Constraint \eqref{eq:budget-constraint} is satisfied: In this case, we have found $\bz$.
\item The weight on $r_{\max}(\by)$ has been exhausted. Let $\widetilde{\by}$ denote the distribution at the point at which $r_{\max}(\by)$ has been exhausted. For this new distribution $\widetilde{\by}$, define $r_{\max}(\widetilde{\by}) = \sup\{r | r \in \supp(\widetilde{\by})\}$, and repeat the process.
\end{enumerate}
To see that this procedure must terminate and return a feasible solution $\bz$, in the worst case we are left with but one $r$ such that $r\fluidn(\unitvec_r) > B$. Let $\widetilde{y}$ be the weight placed on this reward. Since the left-hand side of~\eqref{eq:budget-constraint} is continuously increasing in $\widetilde{y}$ (again, by Lemma~\ref{lem:obvious-fact}), and at $\widetilde{y} = 0$, $\hat{r}(\widetilde{\by})\fluidn(\widetilde{\by}) \leq B$ (since $\rmin\fluidn(\unitvec_{\rmin}) \leq B$), there must exist $\widetilde{y}$ such that~\eqref{eq:budget-constraint} is satisfied. \hfill\Halmos
\end{proof}

\section{Asymptotic optimality of the fluid heuristic}\label{app:fluid-results}

{\subsection{Proof of Proposition~\ref{prop:fluid-ub-for-all-theta}}

Given policy $\varphi$ and $\theta \in \mathbb{N}^+$, let $\randomstatevector^\theta:=\randomstatevector^\theta(\infty)$ denote the steady-state number of agents in the system of each type, and let $\pi^\theta(\varphi)$ be the associated steady-state distribution, if it exists. We moreover define $N^\theta = \sum_{i\in[K]} N^\theta_i$. The upper bound follows from the following lemma.

\begin{lemma}\label{prop:limiting-dist}
For any static policy $\varphi$ defined by reward distribution $\bx$, a unique limiting distribution exists. Thus, $v_\theta(\varphi) = \EE\left[\rev\left(\frac{\scaledrandomstate}{\theta}\right) - \frac{\scaledrandomstate}{\theta}\left(\sum_r r x_r \right)\right].$ Moreover,
$$\scaledrandomstate \sim\normalfont{Pois}\left(\sum_i\frac{\theta\lambda_i}{\sum_r \ell_i(r)x_r}\right).$$
\end{lemma}

\begin{proof}
Since the Markov chain governing $\scaledrandomstate_i(t)$ is irreducible and aperiodic, if a stationary distribution exists, it is the unique limiting distribution. Moreover, by Poisson thinning, $\scaledrandomstate_i$ is Poisson-distributed.
Since the Poisson distribution is uniquely defined by its mean, a stationary distribution exists if and only if the expected number of agents entering the system in each period is equal to the expected number of agents departing from the system, or equivalently 
\begin{align*}
\EE\left[\scaledrandomstate_i\right] = \EE\left[\scaledrandomstate_i\right]\left(1-\sum_r \ell_i(r)x_r\right) + \theta\lambda_i
\implies \EE\left[\scaledrandomstate_i\right] = \frac{\theta\lambda_i}{\sum_r \ell_i(r)x_r} < \infty
\end{align*}
since $\ell_i(r) > 0$ for all $r \in \rewardset$.
Thus, we obtain $\scaledrandomstate\sim\text{Pois}(\sum_i\frac{\theta\lambda_i}{\sum_r \ell_i(r) x_r})$. 

To prove that $v_\theta(\varphi) = \EE\left[\rev\left(\frac{\scaledrandomstate}{\theta}\right) - \frac{\scaledrandomstate}{\theta}\left(\sum_r r x_r \right)\right]$, it suffices to show that $\widehat{\Pi}(\scaledrandomstate) := \rev\left(\frac{\scaledrandomstate}{\theta}\right) - \frac{\scaledrandomstate}{\theta}\left(\sum_r r x_r \right)$ is integrable with respect to $\pi^\theta(\varphi)$. Let $\pi^\theta_n = \PP(\scaledrandomstate = n).$ We have:
\begin{align}\label{eq:integrable-step1}
    \sum_{n=0}^{\infty} \bigg{|} \rev\left(\frac{n}{\theta}\right) - \frac{n}{\theta}\left(\sum_r rx_r\right)\bigg{|} \, \pi^\theta_n < \sum_{n=0}^{\infty} \rev\left(\frac{n}{\theta}\right)\pi^\theta_n + \frac1\theta\left(\sum_r rx_r\right)\sum_{n=0}^{\infty} n \pi^\theta_n
\end{align}
Let $c \geq 0$ by any constant such that $\rev'(c)$ exists. By concavity of $\rev$, $\rev\left(\frac{n}{\theta}\right) \leq \rev(c) + \rev'(c)(\frac{n}{\theta}-c)$. Thus, 
\begin{align*}
    \sum_{n=0}^{\infty} \rev\left(\frac{n}{\theta}\right)\pi^\theta_n \leq  \sum_{n=0}^{\infty} \left(\rev(c) + \rev'(c)\left(\frac{n}{\theta}-c\right)\right)\pi^\theta_n = \rev(c)-cR'(c) + \frac1\theta\rev'(c)\sum_{n=0}^{\infty} n\pi^\theta_n
\end{align*}
Combining this with~\eqref{eq:integrable-step1}, we obtain:
\begin{align*}
    \sum_{n=0}^{\infty} \bigg{|} \rev(n) - n\left(\sum_r rx_r\right)\bigg{|} \, \pi^\theta_n &< \rev(c) - cR'(c) + \frac1\theta\left(\rev'(c) + \sum_r rx_r\right) \sum_{n=0}^{\infty}n\pi^\theta_n \\
    & = \rev(c) - cR'(c) + \frac1\theta\left(\rev'(c) + \sum_r rx_r\right)\EE\left[\scaledrandomstate\right] \\
    &= \rev(c) - cR'(c) + \left(\rev'(c) + \sum_r rx_r\right)\left(\sum_i \frac{\lambda_i}{\sum_r \ell_i(r)x_r}\right) 
    < \infty.
\end{align*}\hfill\Halmos
\end{proof}


\begin{proof}[Proof of Proposition~\ref{prop:fluid-ub-for-all-theta}.]
By \cref{prop:limiting-dist},
\begin{align}
\scaledv(\varphi) = \EE\left[\rev\left(\frac{\scaledrandomstate}{\theta}\right)\right] - \left(\sum_r rx_r\right)\EE\left[\frac{\scaledrandomstate}{\theta}\right] &\leq \rev\left(\EE\left[\frac{\scaledrandomstate}{\theta}\right]\right)- \left(\sum_r rx_r\right)\EE\left[\frac{\scaledrandomstate}{\theta}\right] \label{eq:jensens2}\\
&= \rev\left(\sum_i \frac{\lambda_i}{\sum_r \ell_i(r) x_r}\right)- \left(\sum_r rx_r\right)\sum_i \frac{\lambda_i}{\sum_r \ell_i(r) x_r} \label{eq:poisson-steady-state2} \\
&\leq \fluidopt \label{eq:opt-def2}
\end{align}
where~\eqref{eq:jensens2} follows from concavity of $\rev$,~\eqref{eq:poisson-steady-state2} from \cref{prop:limiting-dist}, and~\eqref{eq:opt-def2} by taking the supremum over all $\xvec \in \simplex^{|\rewardset|}$. Since the upper bound holds for arbitrary $\varphi \in \Phi$, $v^\star_\theta = \sup_{\varphi \in \Phi} \scaledv(\varphi) \leq \optfluidprofit$.
\hfill\Halmos
\end{proof}}

\subsection{Proof of Theorem~\ref{thm:asymptotic-result}}
\begin{proof}
Let $\randomstate^\star$ be the random variable denoting the steady-state number of agents in the system, induced by policy $\varphi^\star$. For ease of notation, let $\Lambda^\star = \sum_i \frac{\lambda_i}{\sum_r \ell_i(r)x_r^\star}$. By \cref{prop:limiting-dist}, $\randomstate^\star \sim$ Poisson$(\theta\Lambda^\star)$.

We first claim that it suffices to establish that 
$\EE\left[\rev\left(\frac{\randomstate^\star}{\theta}\right)\right] \geq \rev\left(\frac{\EEn}{\theta}\right)-\frac{C}{{\theta}},$
since 
\begin{align*}
    \scaledv(\varphi^\star)-\optfluidprofit &= \left(\EE\left[\rev\left(\frac{\randomstate^\star}{\theta}\right)\right]-\left(\sum_r r x_r^\star\right)\frac{\EEn}{\theta}\right) - \left(\rev\left(\frac{\EEn}{\theta}\right)-\left(\sum_r r x_r^\star\right)\frac{\EEn}{\theta}\right) \\
    &= \EE\left[\rev\left(\frac{\randomstate^\star}{\theta}\right)\right]-\rev\left(\frac{\EEn}{\theta}\right).
\end{align*}


 Consider the following three events:
\begin{align*}
 E_1 = \left\{\randomstate^\star \geq \EEn - {c_0}\log\theta\sqrt{\EEn},\,\randomstate^\star \geq 1\right\} \\
E_2 = \left\{ \randomstate^\star < \EEn - {c_0}\log\theta\sqrt{\EEn},\,\randomstate^\star \geq 1 \right\} \\
E_3 = \left\{\randomstate^\star = 0\right\}
\end{align*}
where $c_0 = \frac{\Lambdalb-\varepsilon}{\sqrt{\Lambdalb}}\frac{e}{2}$, for some constant $\varepsilon \in (0,\Lambdalb)$. 

We first decompose the expected revenue under $\varphi^\star$ according to $E_1, E_2, E_3$:
\begin{align}
    \EE\left[\rev\left(\frac{\randomstate^\star}{\theta}\right)\right] &= \EE\left[\rev\left(\frac{\randomstate^\star}{\theta}\right)\bigg{|} E_1\right]\PP(E_1) + \EE\left[\rev\left(\frac{\randomstate^\star}{\theta}\right)\bigg{|} E_2\right]\PP(E_2) + \EE\left[\rev\left(\frac{\randomstate^\star}{\theta}\right)\bigg{|} E_3\right]\PP(E_3) \notag \\
    & \geq \EE\left[\rev\left(\frac{\randomstate^\star}{\theta}\right)\bigg{|} E_1\right]\PP(E_1) + \EE\left[\rev\left(\frac{\randomstate^\star}{\theta}\right)\bigg{|} E_2\right]\PP(E_2). \label{eq:ignore-e3} 
\end{align}
where~\eqref{eq:ignore-e3} follows from non-negativity of $\rev$.

For $i \in \{1,2\}$, let $I_i$ denote the interval corresponding to $E_i$. We leverage the fact that {$\rev$ is twice-continuously differentiable over $\mathbb{R}_{> 0}$}, and apply Taylor's theorem to bound $\rev\left(\frac{\randomstate^\star}{\theta}\right),$ for {$\randomstate^\star \in I_i$, $i \in \{1,2\}$}:
\begin{align}\label{eq:taylor}
\rev\left(\frac{\randomstate^\star}{\theta}\right) = \rev\left(\frac{\EEn}{\theta}\right) + \rev'\left(\frac{\EEn}{\theta}\right)\left(\frac{\randomstate^\star-\EEn}{\theta}\right) + \frac12 \rev''(\eta_i)\left(\frac{\randomstate^\star-\EEn}{\theta}\right)^2
\end{align}
for some {$\eta_i$ between $\randomstate^\star/\theta$ and $\EEn/\theta$, i.e.,} $\eta_i \in \left[\frac1\theta\min\left\{\randomstate^\star,{\EEn}\right\},\frac1\theta\max\left\{{\randomstate^\star},{\EEn}\right\}\right]$.

Plugging the above into~\eqref{eq:ignore-e3}, we obtain:
\begin{align}\label{eq:expected-taylor}
    \EE\left[\rev\left(\frac{\randomstate^\star}{\theta}\right)\right] &\geq \rev\left(\frac{\EEn}{\theta}\right) \left(\PP\left(E_1\right)+\PP\left(E_2\right)\right) + \rev'\left(\frac{\EEn}{\theta}\right)\sum_{i\in\{1,2\}}\EE\left[\frac{{\randomstate^\star-\EEn}}{\theta}\,\bigg{|}\,E_i\right]\PP\left(E_i\right) \notag \\
    &\qquad + \frac{1}{2}\sum_{i\in\{1,2\}}\rev''(\eta_i) \EE\left[\frac{\left(\randomstate^\star-\EEn\right)^2}{\theta^2}\,\bigg{|}\,E_i\right]\PP(E_i) \notag \\
    &= \rev\left(\frac{\EEn}{\theta}\right) \left(1-\PP(E_3)\right) + \rev'\left(\frac{\EEn}{\theta}\right)\sum_{i\in\{1,2\}}\EE\left[\frac{{\randomstate^\star-\EEn}}{\theta}\,\bigg{|}\,E_i\right]\PP\left(E_i\right) \notag \\
    &\qquad + \frac{1}{2}\sum_{i\in\{1,2\}}\rev''(\eta_i) \EE\left[\frac{\left(\randomstate^\star-\EEn\right)^2}{\theta^2}\,\bigg{|}\,E_i\right]\PP(E_i)
\end{align}

Lemmas~\ref{lem:step2-decomposition}-\ref{lem:prob-e3} allow us to complete the {loss} analysis. We defer their proofs to Appendix~\ref{app:asymp-thm-lemmas}.

\begin{lemma}\label{lem:step2-decomposition}
$\rev'\left(\frac{\EEn}{\theta}\right)\sum_{i\in\{1,2\}}\EE\left[\frac{{\randomstate^\star-\EEn}}{\theta}\,\bigg{|}\,E_i\right]\PP\left(E_i\right) \geq 0 \quad \forall \, \theta.$
\end{lemma}

\begin{lemma}\label{lem:e1-analysis}
There exists a constant $c_1 > 0$ such that, conditioned on $E_1$, $$\rev''(x) \EE\left[\frac{\left(\randomstate^\star-\EEn\right)^2}{\theta^2}\,\bigg{|}\,E_1\right]\PP(E_1) \geq -\frac{c_1}{\theta} \qquad \forall \, x\in \left[\frac1\theta\min\left\{\randomstate^\star,{\EEn}\right\},\frac1\theta\max\left\{{\randomstate^\star},{\EEn}\right\}\right].$$
\end{lemma}

\begin{lemma}\label{lem:e2-analysis}
There exists a constant $c_2 > 0$ such that, conditioned on $E_2$, $$\rev''(x) \EE\left[\frac{\left(\randomstate^\star-\EEn\right)^2}{\theta^2}\,\bigg{|}\,E_2\right]\PP(E_2) \geq -\frac{c_2}{\theta} \qquad \forall \, x\in \left[\frac1\theta\min\left\{\randomstate^\star,{\EEn}\right\},\frac1\theta\max\left\{{\randomstate^\star},{\EEn}\right\}\right].$$
\end{lemma}

\begin{lemma}\label{lem:prob-e3}
$\PP(E_3) \leq \frac{1}{e\,\Lambdalb}\cdot\frac1\theta.$
\end{lemma}

Applying the above lemmas to~\eqref{eq:expected-taylor}, we obtain:
\begin{align*}
    \EE\left[\rev\left(\frac{\randomstate^\star}{\theta}\right)\right] &\geq \rev\left(\frac{\EEn}{\theta}\right)\left(1-\frac{1}{e\Lambdalb}\frac{1}{\theta}\right) - \frac12\left(\frac{c_1}{\theta} + \frac{c_2}{\theta}\right).
\end{align*}
Defining $C = \frac{\rev(\Lambdaub)}{e\Lambdalb} + \frac12({c_1+c_2})$, we obtain the theorem.
\hfill\Halmos\end{proof}

\subsubsection{Proofs of Theorem~\ref{thm:asymptotic-result} auxiliary lemmas}\label{app:asymp-thm-lemmas}

\begin{proof}[Proof of Lemma~\ref{lem:step2-decomposition}.]
We have:
\begin{align}\label{eq:step2-decomposition}
    0 = \EE\left[\randomstate^\star-\EEn\right] &= \sum_{i\in\{1,2\}}\EE\left[{\randomstate^\star-\EEn}\bigg{|}E_i\right]\PP\left(E_i\right)
     + \EE\left[{\randomstate^\star-\EEn}\bigg{|}E_3\right]\PP\left(E_3\right) \notag \\
    \implies  \sum_{i\in\{1,2\}}\EE\left[{\randomstate^\star-\EEn}\bigg{|}E_i\right]\PP\left(E_i\right) &= -\EE\left[{\randomstate^\star-\EEn}\bigg{|}E_3\right]\PP\left(E_3\right) = \EEn\PP(E_3) > 0.
\end{align} 

Combining~\eqref{eq:step2-decomposition} with the fact that $\rev'(x) \geq 0$ for all $x$, we obtain the lemma. 
\hfill\Halmos\end{proof}

\begin{proof}[Proof of Lemma~\ref{lem:e1-analysis}.]
For ease of notation, we define $$\mathcal{I}(\randomstate^\star,\theta) := \left[\frac1\theta\min\left\{\randomstate^\star,{\EEn}\right\},\frac1\theta\max\left\{{\randomstate^\star},{\EEn}\right\}\right].$$ Since, conditioned on $E_1$, $\randomstate^\star \geq \theta\Lambda^\star-c_0\log\theta\sqrt{\EEn}$, we have that  $\mathcal{I}(\randomstate^\star,\theta) \subseteq \left[\frac{\theta\Lambda^\star-{c_0}\log\theta\sqrt{\EEn}}{\theta},+\infty\right).$ 

{By assumption, $c_0 \leq \frac{\Lambdalb-\varepsilon}{\sqrt{\Lambdalb}}\frac{e}{2}$, therefore:
\begin{align}
    \frac{\theta\Lambda^\star-{c_0}\log\theta\sqrt{\EEn}}{\theta} = \Lambda^\star-c_0\sqrt{\Lambda^\star}\frac{\log\theta}{\sqrt\theta} \geq \Lambda^\star - \left(\frac{\Lambdalb-\varepsilon}{\sqrt{\Lambdalb}}\frac{e}{2}\right)\sqrt{\Lambda^\star}\frac{\log\theta}{\sqrt\theta} &\geq \Lambda^\star - \left(\frac{\Lambdalb-\varepsilon}{\sqrt{\Lambdalb}}\frac{e}{2}\right)\sqrt{\Lambda^\star}\cdot\frac2e \label{eq:ub-on-logxsqrtx} \\
    &\geq\Lambda^\star - \left(\frac{\Lambdalb-\varepsilon}{\sqrt{\Lambdalb}}\right)\sqrt{\Lambdaub} \geq \varepsilon\label{eq:ub2}
\end{align}
where~\eqref{eq:ub-on-logxsqrtx} follows from the upper bound $\frac{\log\theta}{\sqrt{\theta}} \leq \frac2e$ for all $\theta$, and ~\eqref{eq:ub2} follows from $\Lambda^\star \leq \Lambdaub$, and $\Lambda^\star \geq \Lambdalb$.
}


Thus, $\mathcal{I}(\randomstate^\star,\theta) \subseteq \left[\varepsilon,+\infty\right)$, for some constant $\varepsilon$. Since $\rev$ is twice-continuously differentiable over $\mathbb{R}_{>0}$, there exists $c_1' > 0$ such that $\rev''(x) \geq - c_1'$ for all $x \in \left[\varepsilon,+\infty\right)$, and consequently for all $x \in \mathcal{I}(\randomstate^\star,\theta)$.

We next bound $\EE\left[\frac{\left(\randomstate^\star-\EEn\right)^2}{\theta^2}\right]\PP(E_1).$
\begin{align*}
     &\text{Var}(\randomstate^\star)= \EE\left[(\randomstate^\star-\EEn)^2\right] = \sum_{i=1}^3 \EE\left[(\randomstate^\star-\EEn)^2\bigg{|}E_i\right]\PP(E_i) \\
     \implies &\EE\left[(\randomstate^\star-\EEn)^2\bigg{|}E_1\right]\PP(E_1)\leq \text{Var}(\randomstate^\star) = \theta\Lambda^\star \leq \theta\Lambdaub 
     \implies \EE\left[\frac{\left(\randomstate^\star-\EEn\right)^2}{\theta^2}\bigg{|} E_1
     \right]\PP(E_1) \leq \frac{\Lambdaub}{\theta}.
\end{align*}

Putting this all together, we obtain: $
\rev''(x)\EE\left[\frac{\left(\randomstate^\star-\EEn\right)^2}{\theta^2}\,\bigg{|}\,E_1\right]\PP(E_1) \geq -\frac{c_1'\Lambdaub}{\theta} \quad \forall \, x \in \mathcal{I}(\randomstate^\star,\theta).$
\hfill\Halmos\end{proof}

\begin{proof}[Proof of Lemma~\ref{lem:e2-analysis}.]
By the Poisson
tail bound~\citep{canonne2017short}, 
$\PP(E_2) \leq \exp\left\{-\frac{c_0^2{\log^2\theta}\,\EEn}{2\EEn}\right\} = \theta^{-c_0^2\log\theta/2}.$ Conditioned on $E_2$, we have $\randomstate^\star \in [1,\EEn)$. Thus:
\begin{align*}
&(\randomstate^\star-\EEn)^2 \leq (1-\theta\Lambda^\star)^2 \leq \left(\theta\Lambda^\star\right)^2 \leq \theta^2\Lambdaub^2
\implies \EE\left[\frac{(\randomstate^\star-\EEn)^2}{\theta^2}\,\bigg{|}\,E_2\right] \leq \Lambdaub^2.
\end{align*}

Moreover, {by assumption there exists $\alpha > 0$ such that $\rev''(x) \geq -\theta^{\alpha}$ for all $x \in \left[\frac1\theta,\Lambdaub\right]$. Thus 
\begin{align}
    \rev''(x) \EE\left[\frac{\left(\randomstate^\star-\EEn\right)^2}{\theta^2}\,\bigg{|}\,E_2\right]\PP(E_2) \geq -\Lambdaub^2 \theta^{-c_0^2\log\theta/2}\cdot \theta^{\alpha} \quad \forall\, x \in  \left[\frac1\theta,\Lambdaub\right].
\end{align}

Let $g(\theta) = \theta^{-c_0^2\log\theta/2 + \alpha}.$ We have that $g(\theta) \leq c_2' \theta^{-1}$ for some constant $c_2' > 0$ if and only if $\theta^{-c_0^2\log\theta/2+\alpha+1} \leq c_2' \iff -c_0^2\log\theta/2+\alpha+1 \leq \log c_2' \iff c_0^2 \geq \frac{2(1+\alpha-\log c_2')}{\log\theta}$. For {$c_2' > e^{1+\alpha}$}, this holds for all $c_0 > 0$, $\theta > 1$.}

We conclude the proof of the lemma by observing that, conditioned on $E_2$, $$\left[\frac1\theta\min\left\{\randomstate^\star,{\EEn}\right\},\frac1\theta\max\left\{{\randomstate^\star},{\EEn}\right\}\right] = \left[\frac{\randomstate^\star}{\theta},\frac{\EEn}{\theta}\right]\subseteq \left[\frac1\theta,\Lambda^\star\right] \subseteq\left[\frac1\theta,\Lambdaub\right].$$  Thus, the bound holds for all $x \in \left[\frac1\theta\min\left\{\randomstate^\star,{\EEn}\right\},\frac1\theta\max\left\{{\randomstate^\star},{\EEn}\right\}\right]$.
\hfill\Halmos\end{proof}

\begin{proof}[Proof of Lemma~\ref{lem:prob-e3}.]
By definition,
$\PP(E_3) = \PP(\randomstate^\star = 0) = e^{-\theta\Lambda^\star} \leq \frac1e \cdot \frac{1}{\theta\Lambda^\star} \leq \frac{1}{e\,\Lambdalb}\cdot\frac1\theta.$
\hfill\Halmos\end{proof}

\chedit{
\subsection{Performance of static policies in small-market regimes}\label{app:small_market_numerics}

Our results establish fast convergence of the fluid heuristic to the fluid upper bound in a large-market regime (i.e., as $\theta\to\infty$). While such regimes are meaningful within the context of, e.g., CreditKarma, where the customer base is on the order of millions, a natural question is whether one can bridge this gap when $\theta$ is small. In this section we explore the gap between the fluid optimum and the performance of the fluid heuristic in small- and moderate-sized markets. 

To do so, we first consider an instance where $\Xi = \{15,40,60\}$, and the revenue function is newsvendor-like, i.e. $R(N) = 100\min\{N,5\}$. We assume $K = 1$, with $\lambda_1 = 1$ and $\ell_1(r) = \begin{cases}
0.8 \quad \text{if } r = 15 \\
0.5 \quad \text{if } r = 40 \\
0.1 \quad \text{if } r = 60
\end{cases}.
$ In addition to the fluid heuristic, for $\theta \in \{1,5,10,\ldots,100\}$, we enumerate over all static policies via grid search to identify the optimal static policy for the {\it stochastic} system, i.e., the distribution that maximizes $$\lim_{T\to\infty}\frac1T\sum_{t=1}^T\mathbb{E}\left[R\left(\frac{N^{\theta}(t)}{\theta}\right)-\left(\sum_r rx_r\right)\cdot\frac{N^\theta(t)}{\theta}\right].$$
\cref{fig:small-market} compares the performance of these two policies to the fluid optimum.

\begin{figure}[h]
     \centering{\includegraphics[scale=0.6]{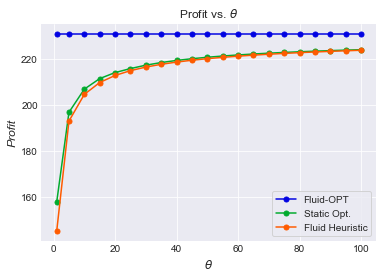}}
    \caption{Fluid heuristic and static optimum performance versus fluid upper bound in small-$\theta$ regimes}
    \label{fig:small-market}
\end{figure}

For $\theta = 1$, we observe a significant gap (at least 30\% relative difference) between the fluid upper bound, and both the static optimum and fluid heuristic. Moreover, there is approximately an 8\% relative difference between the static optimum and the fluid heuristic. These observations lead us to two conclusions: $(i)$ for truly small-market regimes, there may indeed be room to improve upon static policies, and $(ii)$ within the space of static policies, it is important to take into account the variability of the underlying arrival distribution. Indeed, while the fluid heuristic ensures that 5 agents are in the system in expectation, the static optimum for small $\theta$ induces slightly less than 4 agents in the system in expectation. Essentially, the optimal static policy aims to avoid paying a premium for having $>5$ agents that do not increase the revenue. 

This difference is all but erased for $\theta = 5$, with about a 2\% relative gap between the static optimum and the fluid heuristic due to Poisson concentration. Still, a significant gap remains between these two policies and the fluid upper bound. This gap also shrinks for moderate $\theta$ (e.g., 3\% relative difference at $\theta = 100$).

{We next investigate the convergence under two strictly concave revenue functions. \cref{fig:strict-concave} illustrates our results for $R(N) = 100\min\{N^{1-\epsilon},5N^{\epsilon}\}$, for $\epsilon \in \{0.1,0.2\}$. Here we find that for small values of $\epsilon$, when $\theta$ is small, the gap between the optimal static policy and the fluid heuristic is even larger than in the case of Figure \ref{fig:small-market}. However, we observe significantly faster convergence to the fluid optimum in these settings. Indeed, for $\epsilon = 0.2$, while this gap is negligible as of $\theta = 5$, with the static optimum and the fluid heuristic achieving identical performance for all values of $\theta$.

\begin{figure}
     \centering
     \subfloat[$\epsilon=0.1$
     \label{fig:small-theta-0.1}]{\includegraphics[width = 0.5\textwidth]{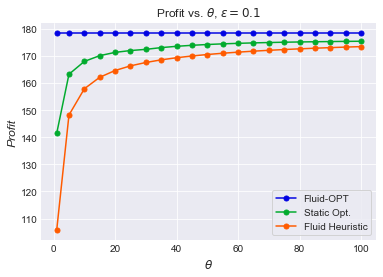}}
          \subfloat[$\epsilon = 0.2$
     \label{fig:small-theta-0.2}]{\includegraphics[width = 0.5\textwidth]{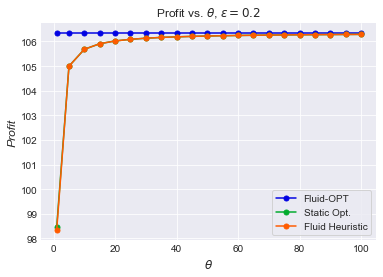}}
    \caption{$\rev(N) = 100\min\{N^{1-\epsilon},5N^{\epsilon}\}$.}
    \label{fig:strict-concave} 
\end{figure}

}
Taken together, our numerical findings provide additional evidence for the appropriateness of our asymptotic analysis. Despite the fluid heuristic not always giving us the optimal static policy, we find that even in the worst cases, at $\theta\approx 100$ the fluid system becomes a good approximation of the stochastic model.
}

\chedit{
\section{Special cases}\label{apx:special-cases}

\subsection{On the convexity of the departure function}

\subsubsection{Proof of \cref{thm:main-theorem2}}
\begin{proof}
The proof uses a similar strategy as that of Theorem~\ref{thm:main-theorem}. In particular, we rely on the following lemma, analogous to Lemma~\ref{prop:profit-reduction}. We leave the almost identical proof of Lemma~\ref{lem:dispersion-reduction} to the reader.
\begin{lemma}\label{lem:dispersion-reduction}
Suppose that, for all instances of~\ref{eq:supply-opt}, there exists an optimal solution $\nf{\supplyoptsol}$ to~\ref{eq:supply-opt} which is minimally (resp. maximally) dispersed. Then, there exists an optimal solution $\nf{\bx}^\star$ to~\ref{eq:fluid-opt} which is minimally (resp. maximally) dispersed. 
\end{lemma}

Thus, it suffices to show that the solution to~\ref{eq:supply-opt} has the correct dispersion structure, depending on the convexity of $\{\ell_i\}$. We prove this for strictly {concave} departure probabilities, and leave the symmetric proof of the setting with strictly {convex} departure probabilities to the reader. 

Consider any non-maximally-dispersed feasible solution $\bx$ to~\ref{eq:supply-opt}. Let $r_2> r_3$ be such that $\supp(\bx) = \{r_2,r_3\}$, and let $x$ denote the weight placed on $r_2$. Consider reward $r_1 \in \rewardset$ such that $r_1 > r_2$.\footnote{A symmetric argument holds if $r_2= \rmax$ by assuming $r_1 < r_3$.} We construct a feasible solution $\by$ to~\ref{eq:supply-opt} such that $\supp(\by) = \{r_1, r_3\}$ which weakly improves upon $\bx$. Specifically, we define $\by$ to be such that $\avgreward(\by) = \avgreward(\bx)$, and let $y$ denote the weight placed on $r_1$. Solving for $y$, we have:
\begin{align*}
r_1 y + r_3(1-y) = r_2x + r_3 (1-x) \iff y = \frac{r_2-r_3}{r_1-r_3}x.
\end{align*}
Under $\by$, the expected departure probability of type $i$ agents{, $\hat{\ell}_i(\by)$,} is
\begin{align*}
\hat{\ell}_i(\by) &= \ell_{i}(r_1)\cdot \frac{r_2-r_3}{r_1-r_3}x + \ell_{i}(r_3)\left(1- \frac{r_2-r_3}{r_1-r_3}x\right) \\
&= \hat{\ell}_i(\bx) - \ell_{i}(r_2) x - \ell_{i}(r_3) \left(1-x\right) + \ell_{i}(r_1)\, \frac{r_2-r_3}{r_1-r_3}x + \ell_{i}(r_3)\left(1- \frac{r_2-r_3}{r_1-r_3}x\right) \\
&= \hat{\ell}_i(\bx) + x \left(-\ell_{i}(r_2) + \ell_{i}(r_1)\,\frac{r_2-r_3}{r_1-r_3}+\ell_{i}(r_3)\left(1-\frac{r_2-r_3}{r_1-r_3}\right) \right) \\
&< \hat{\ell}_i(\bx).
\end{align*}
where the inequality follows from strict concavity of $\ell_i$ (Definition~\ref{def:ccv}).

Thus, $\fluidn_i(\by) = \frac{\lambda_i}{\hat{\ell}_i(\by)} > \frac{\lambda_i}{\hat{\ell}_i(\bx)} = \fluidn_i(\bx)$ for all $i \in [K]$, and $\fluidn(\by) > \fluidn(\bx)$. As argued in the proof of Lemma~\ref{lem:support-reduction-maintains-invariants}, the budget constraint~\eqref{eq:budget-constraint} has been violated; we can however use the same redistribution procedure, continuously moving weight from $r_1$ to $r_3$ until $\by$ becomes feasible. This procedure strictly decreases $\avgreward(\by)$, thus strictly improving upon $\bx$.\hfill\Halmos
\end{proof}
}

\chedit{
\subsection{Linearized S-Shape}

\subsubsection{Proof of \cref{prop:noisy-opt}}

\begin{proof}
That $\bx^\star(\epsilon)$ randomizes between $\rmin$ and $v+\epsilon$ follows from the fact that, when $K = 1$, it is never optimal for the decision-maker to pay out a reward $r > v+\epsilon$. As a result, it is without loss of generality to consider the departure probability function on the domain $[\rmin, v+\epsilon]$, over which it is strictly concave. By \cref{thm:main-theorem2}, then, $\bx^\star(\epsilon)$ is maximally dispersed.

We moreover have that $x^\star(\epsilon) < 1$, as $\ell(v+\epsilon) = 0$ implies that deterministically paying out $v+\epsilon$ would induce infinitely many agents in the system. Since $\rev$ is strictly concave, this cannot be optimal.

We now prove the threshold structure. For fixed $\epsilon \leq \min \{v-\rmin,\rmax-v\}$, let $\widetilde{N}(\bx;\epsilon)$ denote the number of agents in the deterministic relaxation under reward distribution $\bx$ such that $\supp(\bx) = \{\rmin, v+\epsilon\}$. Moreover, let $x$ denote the weight placed on $v+\epsilon$ under $\bx$. We have:
\begin{align*}
    \widetilde{N}(\bx; \epsilon) &= \frac{\lambda}{\widehat{\ell}(\bx)} = \frac{\lambda}{\ell(v+\epsilon)x + \ell(\rmin)\cdot (1-x)} = \frac{\lambda}{1-x},
\end{align*}
since $\ell(v+\epsilon) = 0$ and $\ell(\rmin) = 1$. Thus, the profit $\widetilde{\Pi}(\bx; \epsilon)$ induced by $\bx$ is given by:
\begin{align*}
&\widetilde{\Pi}(\bx;\epsilon) = \rev\left(\frac{\lambda}{1-x}\right) - \left((v+\epsilon)x + \rmin(1-x)\right)\frac{\lambda}{1-x} \\
\implies &\frac{\partial\widetilde{\Pi}}{\partial x} = \frac{\lambda}{(1-x)^2}\left[\rev'\left(\frac{\lambda}{1-x}\right) -(v+\epsilon)\right] \quad \forall \, x \in [0,1).
\end{align*}
If $\rev'\left(\frac{\lambda}{1-x}\right) \leq v+\epsilon$ for all $x \in [0,1)$,  $\widetilde{\Pi}(\bx;\epsilon)$ is {decreasing} for all $x$, and setting $x = 0$ is optimal. Noting that $\rev'\left(\frac{\lambda}{1-x}\right)$ is {\it decreasing} in $x$, since $\frac{1}{1-x}$ is increasing in $x$ and $\rev$ is strictly concave, we have that $\rev'\left(\frac{\lambda}{1-x}\right) \leq \rev'(\lambda)$ for all $x \in [0,1)$. Thus, for all $\epsilon$ such that $\rev'(\lambda) \leq v+\epsilon$ (i.e., $\epsilon \geq \rev'(\lambda)-v$), $x^\star(\epsilon) = 0$. Note that if $\rev'(\lambda)-v\leq 0$, this holds for all $\epsilon > 0$, which proves the first part of the proposition.

Suppose now that $\rev'(\lambda) > v$. As before, if $\epsilon \geq \rev'(\lambda)-v =: \epsilon_0$, $x^\star(\epsilon) = 0$. If $\epsilon < \rev'(\lambda)-v$,
since $\rev'\left(\frac{\lambda}{1-x}\right)$ is continuously decreasing in $x$, there must be $x < 1$ such that $\epsilon = \rev'\left(\frac{\lambda}{1-x}\right)-v$, i.e., the first-order condition is satisfied. (This also follows from the fact that $\lim_{x\to 1}\rev'\left(\frac{\lambda}{1-x}\right) = 0$, again since $\rev$ is strictly concave and increasing.)

To see that $x^\star(\epsilon)$ is decreasing for $\epsilon < \rev'(\lambda)-v$, note that the solution to $\rev'\left(\frac{\lambda}{1-x}\right) = v+\epsilon$ decreases with $\epsilon$, again since $\rev'\left(\frac{\lambda}{1-x}\right)$ is continuously decreasing in $x$.
\hfill\Halmos\end{proof}

\medskip

\subsubsection{Proof of \cref{prop:noisy-profit}}
\begin{proof}
Recall,
\begin{align*}
    \widetilde{\Pi}(\bx^\star(\epsilon);\epsilon) &= \rev\left(\frac{\lambda}{1-x^\star(\epsilon)}\right) - \left((v+\epsilon)x^\star(\epsilon) +\rmin(1-x^\star(\epsilon))\right)\frac{\lambda}{1-x^\star(\epsilon)} \\ 
    &= \rev\left(\frac{\lambda}{1-x^\star(\epsilon)}\right) - \left(v+\epsilon-\rmin\right)\frac{\lambda x^\star(\epsilon)}{1-x^\star(\epsilon)} - \rmin \frac{\lambda}{1-x^\star(\epsilon)}
\end{align*}
For $\epsilon > \epsilon_0$, $x^\star(\epsilon) = 0$, by \cref{prop:noisy-opt}. Thus $ \widetilde{\Pi}(\bx^\star(\epsilon);\epsilon) = \rev(\lambda) - \rmin\lambda$, a constant.
For $\epsilon \leq \epsilon_0$, we have:
{
\begin{align*}
    \frac{d}{d\epsilon}\widetilde{\Pi}(\bx^\star(\epsilon);\epsilon) &= \frac{dx^\star(\epsilon)}{d\epsilon}  \frac{\lambda}{(1-x^\star(\epsilon))^2}\rev'\left(\frac{\lambda}{1-x^\star(\epsilon)}\right) - \lambda\frac{x^\star(\epsilon)}{1-x^\star(\epsilon)} \\
    &\qquad-(v+\epsilon-\rmin)\frac{\lambda}{(1-x^\star(\epsilon))^2}\frac{dx^\star(\epsilon)}{d\epsilon}-\rmin\frac{\lambda}{(1-x^\star(\epsilon))^2}\frac{dx^\star(\epsilon)}{d\epsilon} \\
    &=\frac{dx^\star(\epsilon)}{d\epsilon}\frac{\lambda}{(1-x^\star(\epsilon))^2}\left(\rev'\left(\frac{\lambda}{1-x^\star(\epsilon)}\right)-(v+\epsilon)\right) -\lambda\frac{x^\star(\epsilon)}{1-x^\star(\epsilon)}
\end{align*}
}
By \cref{prop:noisy-opt}, $\frac{dx^\star(\epsilon)}{d\epsilon} \leq 0$ for all $\epsilon$. Moreover, $\rev'\left(\frac{\lambda}{1-x^\star(\epsilon)}\right)-(v+\epsilon) = 0$ holds by the first-order condition, satisfied at $x^\star(\epsilon)$ for all $\epsilon \leq \epsilon_0$. Thus, $\frac{d}{d\epsilon}\widetilde{\Pi}(\bx^\star(\epsilon);\epsilon)\leq 0$ for all $\epsilon \leq \epsilon_0$. 
\hfill\Halmos\end{proof}

}

\chedit{
\section{Extension of \cref{thm:static-policies-are-opt-for-one-type} to model with two-period memory}\label{apx:memory-extension}

Consider now a setting in which agents make the decision to stay in or leave the system based on the two most recent rewards received. For a type $i$ agent, let $\ell_i(r_1, r_2)$ denote the probability that she leaves the system having received a reward $r_1$ in the current period, and $r_2$ in the last period. We assume $\ell_i(r_{\max}, r_{\max}) > 0$. Since agents make their decision based on the last two rewards, we assume that they remain in the system in the first period they arrive, independent of the reward received. 

The decision-maker optimizes over the space of joint distributions of rewards for each two consecutive periods, i.e., $x_{r_1,r_2}(t, t-1) := \mathbb{P}\left(r(t) = r_1, r(t-1) = r_2\right)$, for all $r_1, r_2, t$, where we abuse notation and let $r(t)$ denote the reward paid out in period $t$. Given a sequence of reward distributions, the number of type $i$ agents in the deterministic system satisfies the following inductive relation:
\begin{align}\label{eq:inductive-memory}
    \fluidn_i(t+1) = \fluidn_i(t) + \lambda_i - (\fluidn_i(t)-\lambda_i)\sum_{r_1,r_2}\ell_i(r_1,r_2)x_{r_1,r_2}(t,t-1), \quad \forall\, t \geq 1
\end{align}

Our notion of group fairness, adapted to the setting with two-period memory, is as follows:
\begin{definition}[Two-period memory group-fair policy]\label{def:memory_fair_policy}
Policy $\varphi$ is \emph{group-fair} if, for all $\delta > 0$, there exists $\tau_0 \in \mathbb{N}^+$ such that for all $\tau > \tau_0$:
\begin{align}\label{eq:memory-fair-policy}
   \sum_{r_1,r_2} \bigg{\lVert}\frac{1}{\sum_{t=t'}^{t'+\tau}(\widetilde{N}^\varphi_i(t)-\lambda_i)}\sum_{t=t'}^{t'+\tau} (\widetilde{N}^\varphi_i(t)-\lambda_i)x_{r_1,r_2}(t,t-1) - \frac{1}{\sum_{t=t'}^{t'+\tau}(\widetilde{N}^\varphi_j(t)-\lambda_j)}\sum_{t=t'}^{t'+\tau} (\widetilde{N}^\varphi_j(t)-\lambda_j)x_{r_1,r_2}(t,t-1) \bigg{\rVert}_1\\ < \delta \notag   \forall \, t' \in \mathbb{N}^+, \, \forall \, i,j \in [K]. 
\end{align}
\end{definition}
Note that the group-fairness definition here is with respect to $\widetilde{N}^\varphi_i(t)-\lambda_i$, the number of agents who have been in the system for at least two periods. {One could similarly include a separate constraint for agents who just joined the system by considering the marginal distribution of rewards in each period; including such a constraint, however, trivially reduces the space of policies to static policies, given the fact that arrival rates are stationary.} 
The following theorem shows that the optimality of static policies within the space of all fair policies is robust to agent memory.
\begin{theorem}\label{thm:static-policies-are-opt-memory}
There exists an optimal fair policy that is static, i.e., {$\mathbb{P}(r(t-1)=r_1) = \mathbb{P}(r(t) = r_1)$ for all $t$.} 
\end{theorem} 

\begin{proof}
{We show that there exists an optimal fair policy such that $x_{r_1,r_2}(t,t-1) = x_{r_1,r_2}(t+1,t)$. Summing over all $r_2$ gives the claim.}

As for the memoryless setting, the proof is constructive. Fix an arbitrary fair policy, and consider the static policy that chooses an arbitrary type $i$, and is defined by reward distribution:
\[\widehat{x}_{r_1,r_2} = \lim_{T\to\infty}\frac{\sum_{t=1}^T (\fluidn_i(t)-\lambda_i)x_{r_1,r_2}(t,t-1)}{\sum_{t=1}^T(\fluidn_i(t)-\lambda_i)} \quad \forall \ r_1,r_2 \in \Xi.\]
Let $\widehat{N}_i(t)$ denote the number of type $i$ agents in the deterministic system under $\widehat{x}_{r_1,r_2},$ and define $\widehat{N}_i = \lim_{T\to\infty}\frac1T\sum_{t=1}^T\widehat{N}_i(t)$. It is easy to verify that
\[\widehat{N}_i = \lambda_i + \frac{\lambda_i}{\sum_{r_1,r_2}\ell_i(r_1,r_2)\widehat{x}_{r_1,r_2}}.\]

Finally, we abuse notation and let $x_r(t), \widehat{x}_r$ denote the marginal probability that each respective policy distributes reward $r$ in period $t$.

{By the same arguments as those used in the proof of \cref{thm:static-policies-are-opt-for-one-type},} the difference in the two policies' long-run average profit $\Delta$ is upper bounded by:
{
\begin{align*}
\Delta &\leq L|\lim_{T\to\infty}\frac1T\sum_{t=1}^T\fluidn(t) -\widehat{N}| + \sum_{r}r\left(\widehat{x}_r\widehat{N} - \lim_{T\to\infty}\frac1T\sum_{t=1}^Tx_r(t)\fluidn(t)\right).
\end{align*}
}

The following lemma establishes that the static policy induces the same long-run average number of agents in the system as the original policy. We defer its proof to the end of the section.
\begin{lemma}\label{lem:mem-step-1}
For all $j\in[K]$, $\lim_{T\to\infty}\frac1T\sum_{t=1}^T\fluidn_j(t) = \widehat{N}_j$.
\end{lemma}
Hence, it suffices to bound the difference in average rewards paid out in each period. By \cref{lem:mem-step-1} and the static construction, for all $r \in \Xi$:
\begin{align*}
\widehat{N}\widehat{x}_r = \left(\lim_{T\to\infty}\frac1T\sum_{t=1}^T\fluidn(t)\right)\cdot\lim_{T\to\infty}\frac{\sum_{t=1}^T(\fluidn_i(t)-\lambda_i)x_r(t)}{\sum_{t=1}^T(\fluidn_i(t)-\lambda_i)}.
\end{align*}
We also have:
\begin{align*}
\lim_{T\to\infty}\frac1T\sum_{t=1}^T\fluidn(t)x_r(t) &= \left(\lim_{T\to\infty}\frac{\sum_{t=1}^T\fluidn(t)}{T}\right)\cdot\left(\lim_{T\to\infty}\frac{\sum_{t=1}^T\fluidn(t)x_r(t)}{\sum_{t=1}^T\fluidn(t)}\right).
\end{align*}
Putting these two together, we obtain:
\begin{align}
\widehat{N}\widehat{x}_r- \lim_{T\to\infty}\frac1T\sum_{t=1}^T\fluidn(t)x_r(t) &= \left(\lim_{T\to\infty}\frac{\sum_{t=1}^T\fluidn(t)}{T}\right)\cdot {\left(\lim_{T\to\infty}\frac{\sum_{t=1}^T(\fluidn_i(t)-\lambda_i)x_r(t)}{\sum_{t=1}^T(\fluidn_i(t)-\lambda_i)}-\lim_{T\to\infty}\frac{\sum_{t=1}^T\fluidn(t)x_r(t)}{\sum_{t=1}^T\fluidn(t)}\right)} \notag \\
&\leq \left(\lim_{T\to\infty}\frac{\sum_{t=1}^T\fluidn(t)}{T}\right)\cdot \underbrace{\left(\lim_{T\to\infty}\frac{\sum_{t=1}^T(\fluidn_i(t)-\lambda_i)x_r(t)}{\sum_{t=1}^T(\fluidn_i(t)-\lambda_i)}-\lim_{T\to\infty}\min_j\frac{\sum_{t=1}^T\fluidn_j(t)x_r(t)}{\sum_{t=1}^T\fluidn_j(t)}\right)}_{(I)},
\end{align}
where the first inequality follows from identical arguments as those used for \cref{lem:endo-bound-2}. We have the following fact, whose proof we defer to the end of the section.
\begin{lemma}\label{lem:necessary-fact}
For any $T \in \mathbb{N}, j \in [K]$,
$\frac{\sum_{t=1}^T(\fluidn_j(t)-\lambda_j)x_r(t)}{\sum_{t=1}^T(\fluidn_j(t)-\lambda_j)} \leq \frac{\sum_{t=1}^T\fluidn_j(t)x_r(t)}{\sum_{t=1}^T\fluidn_j(t)}$.
\end{lemma}
We use this to upper bound $(I)$ as follows:
\begin{align}\label{eq:last-step-memory}
(I) 
&\leq \lim_{T\to\infty}\left(\frac{\sum_{t=1}^T(\fluidn_i(t)-\lambda_i)x_r(t)}{\sum_{t=1}^T (\fluidn_i(t)-\lambda_i)}-\min_j\frac{\sum_{t=1}^T(\fluidn_j(t)-\lambda_j)x_r(t)}{\sum_{t=1}^T(\fluidn_j(t)-\lambda_j)}\right).
\end{align}

\cref{def:memory_fair_policy} implies that, for all $\delta > 0$, there exists large enough $\tau_0$ such that, for all $T > \tau_0$:
\begin{align}
   \bigg{|}\frac{1}{\sum_{t=1}^{T}(\widetilde{N}_i(t)-\lambda_i)}\sum_{t=1}^{T} (\widetilde{N}_i(t)-\lambda_i)x_{r}(t) - \frac{1}{\sum_{t=1}^{T}(\widetilde{N}_j(t)-\lambda_j)}\sum_{t=1}^{T} (\widetilde{N}_j(t)-\lambda_j)x_r(t) \bigg{|} < \delta \quad \, \forall \, i,j \in [K], r \in \Xi. 
\end{align}
Applying this to \eqref{eq:last-step-memory}, and taking $\delta\to 0$ we obtain the result.
\hfill\Halmos
\end{proof}

{\begin{remark}
\revedit{The two-period memory assumption was not crucial here. In fact, we conjecture that these arguments extend to a $m$-period memory model in which the platform's lever is the space of all joint distributions over $m$ consecutive periods. However, with $m=2$, at a high level, it is easier to observe that the right construction is the weighted long-run average of the joint distribution of rewards in the past $m$ periods. In contrast, for general $m$, the notational overhead becomes burdensome and obfuscates the technical insight.}
\end{remark}}

\subsection{Auxiliary lemmas}
\begin{proof}[Proof of \cref{lem:mem-step-1}.]
Summing \eqref{eq:inductive-memory} from $t = 1$ to $T$, we have:
\begin{align*}
&\fluidn_j(T+1)-\fluidn_j(1) = \lambda_jT - \sum_{t=1}^T(\fluidn_j(t)-\lambda_j)\sum_{r_1,r_2}\ell_j(r_1,r_2)x_{r_1,r_2}(t,t-1) \\
\implies &\lim_{T\to\infty} \frac{\sum_{t=1}^T(\fluidn_j(t)-\lambda_j)\sum_{r_1,r_2}\ell_j(r_1,r_2)x_{r_1,r_2}(t,t-1)}{T} = \lambda_j\\ 
\implies &\lim_{T\to\infty} \frac{\sum_{t=1}^T(\fluidn_j(t)-\lambda_j)}{T} \cdot \frac{\sum_{t=1}^T(\fluidn_j(t)-\lambda_j)\sum_{r_1,r_2}\ell_j(r_1,r_2)x_{r_1,r_2}(t,t-1)}{\sum_{t=1}^T(\fluidn_j(t)-\lambda_j)} = \lambda_j \\
\implies &\lim_{T\to\infty} \frac{\sum_{t=1}^T(\fluidn_i(t)-\lambda_i)}{T} \cdot \sum_{r_1,r_2}\ell_j(r_1,r_2)\lim_{T\to\infty}\frac{\sum_{t=1}^T(\fluidn_i(t)-\lambda_i)x_{r_1,r_2}(t,t-1)}{\sum_{t=1}^T(\fluidn_i(t)-\lambda_i)} = \lambda_j
\end{align*}
where the final implication applies the limiting version of the group-fairness constraint \eqref{eq:memory-fair-policy}. Using the definition of $\widehat{x}_{r_1,r_2}$ and re-arranging, we obtain:
\begin{align*}
&\lim_{T\to\infty} \frac{\sum_{t=1}^T\fluidn_j(t)}{T} = \lambda_j + \frac{\lambda_j}{\sum_{r_1,r_2}\ell_j(r_1,r_2)\widehat{x}_{r_1,r_2}} = \widehat{N}_j.
\end{align*}
 \hfill\Halmos
\end{proof}

\begin{proof}[Proof of \cref{lem:necessary-fact}.]
For $j \in [K]$, let $Z_j = \left(\sum_{t=1}^T(\fluidn_j(t)-\lambda_j)\right)\cdot\sum_{t=1}^T\fluidn_j(t)$. We have:
\begin{align*}
\frac{\sum_{t=1}^T(\fluidn_j(t)-\lambda_j)x_r(t)}{\sum_{t=1}^T(\fluidn_j(t)-\lambda_j)} - \frac{\sum_{t=1}^T \fluidn_j(t)x_r(t)}{\sum_{t=1}^T \fluidn_j(t)} 
&=\sum_{t=1}^Tx_r(t)\left[\frac{(\fluidn_j(t)-\lambda_j)}{\sum_{t'=1}^T(\fluidn_j(t')-\lambda_j)} - \frac{ \fluidn_j(t)}{\sum_{t'=1}^T \fluidn_j(t')}\right]\\
&= \frac{1}{Z_j}\sum_{t=1}^Tx_r(t)\bigg[(\fluidn_j(t)-\lambda_j)\left(\sum_{t'=1}^T\fluidn_j(t')\right)\\
&\qquad\qquad -\fluidn_j(t)\left(\sum_{t'=1}^T(\fluidn_j(t')-\lambda_j)\right)\bigg] \\
&=\frac{\lambda_j}{Z_j}\sum_{t=1}^Tx_r(t)\left(-\left(\sum_{t'=1}^T\fluidn_j(t')\right)+\fluidn_j(t)\right)\\
&=-\frac{\lambda_j}{Z_j}\sum_{t=1}^Tx_r(t)\left(\sum_{t'\neq t}\fluidn_j(t')\right) \leq 0.
\end{align*}\hfill\Halmos
\end{proof}
}

\end{document}